\documentclass[journal]{IEEEtran}

\usepackage{extramath}
\usepackage{amsthm}
\newtheorem{definition}{Definition}

\newtheorem{lemma}{Lemma}

\usepackage[font=footnotesize,labelsep=period]{caption}
\usepackage[font=footnotesize,labelformat=parens,labelsep=space]{subcaption}
\usepackage{pgfplots}
\pgfplotsset{compat=1.16}
\usepgfplotslibrary{groupplots}
\pgfplotsset{table/search path={figures}}

\usepackage{url}
\usepackage{cite}

\begin{document}
\title{Sparse Approximation via Polynomial Equations}

\author{Matija~Tomi\'c,~
        Raphaël~Widdershoven,~
        and~Lieven~De~Lathauwer~
\thanks{{M. Tomi\'c was with the Department of Electrical Engineering (ESAT), KU Leuven, Leuven, Belgium. He is now with the Department of Mathematics, University of British Columbia, Vancouver, BC, Canada.}}%
\thanks{{R. Widdershoven was with the Department of Electrical Engineering (ESAT), KU Leuven, Leuven, Belgium.}}%
\thanks{{L. De Lathauwer is with the Department of Electrical Engineering (ESAT), KU Leuven, Leuven, Belgium, and also with Leuven.AI -- KU Leuven institute for AI, B-3000 Leuven, Belgium.}}
}

\maketitle

\begin{abstract}
We consider the problem of finding sparse solutions of an underdetermined linear system $\mathbf{Ax}=\mathbf{b}$. In contrast to conventional approaches based on greedy algorithms or convex relaxation, we reformulate sparse approximation as a structured system of polynomial equations and connect with the literature on tensor methods. We develop an eigenvalue decomposition based method that formally guarantees recovery of all sparse solutions if there is more than one. We also develop two optimization-based methods achieving favorable computational complexity. The new methods allow explicit control of the target sparsity. Numerical experiments illustrate the performance and compare to basis pursuit (denoising) and orthogonal matching pursuit.
\end{abstract}


\section{Introduction}

\IEEEPARstart{M}{any} fundamental problems in signal processing and applied mathematics can be formulated as sparse approximation problems. The wide applicability of this framework is demonstrated across diverse domains, such as compressed sensing \cite{donoho}, statistics \cite{lasso}, dictionary learning \cite{dict_learning, ksvd}, image processing \cite{elad, denoising, super_resolution, face_recog}, and data-driven discovery of governing equations \cite{sindy}. 

From the viewpoint of signal processing, the goal of sparse approximation is to represent a measured signal as a linear combination of a few elementary signals. Sparse approximation aims to find sparse solutions \( \mathbf{x} \) to possibly underdetermined systems \( \mathbf{A} \mathbf{x} = \mathbf{b} \), where \( \mathbf{A} \in \mathbb{R}^{m \times n} \) with \( m < n \). The sparsity level \( s \) denotes the number of nonzero entries in \( \mathbf{x} \).

A brute-force search over all \( \binom{n}{s} \) subsets of columns in \( \mathbf{A} \) is computationally intractable for large \( n \), and the problem is NP-hard \cite{sparse_np_hard1, sparse_np_hard2}. Practical algorithms, such as greedy methods (e.g., Orthogonal Matching Pursuit (OMP) \cite{omp1, omp2, omp1950}) and convex relaxations (e.g., Basis Pursuit (BP), Basis Pursuit Denoising (BPDN), \cite{basis_pursuit}), offer approximate solutions with theoretical guarantees under specific conditions; see \cite{elad, FoucartRauhut2013} and references therein.

{In this paper, we present an alternative approach based on solving systems of polynomial equations. We first formulate sparse approximation as a polynomial system (Section~\ref{sec:connection_sparse_poly}) and explain how tensors represent and solve such systems (Section~\ref{sec:tensors}). The sparse solutions appear as the roots of a polynomial system, all of which {can be} recovered simultaneously by the Macaulay construction (Section~\ref{sec:polyeq}). Additionally, we propose two computationally tractable {optimization-based} alternatives, SESP-D and SESP-P, and analyze their complexity (Section~\ref{sec:ls_cp_based}). Finally, we {experimentally verify their performance and }compare them to BP/BPDN and OMP~(Section~\ref{sec:comparison}).}

\label{sec:intro}
{\emph{Notation: }}
A tensor generalizes scalars, vectors, and matrices to higher orders \cite{kolda_overview,nikos_overview}. Scalars are denoted by lowercase letters $(a)$, vectors by bold lowercase letters $(\mathbf{a})$, matrices by bold uppercase letters $(\mathbf{A})$, and tensors by calligraphic letters $(\mathcal{A})$. Matrix transpose, inverse, and pseudoinverse are written as $\cdot^\top$, $\cdot^{-1}$, and $\cdot^{\dagger}$, respectively. 
The $i$-th column of $\mathbf{A}$ is denoted by $\mathbf{a}_i$, thus $\mathbf{A} = [\mathbf{a}_1~\mathbf{a}_2~\cdots]$. A diagonal matrix with $\mathbf{a}$ on the diagonal is written as $\text{diag}(\mathbf{a})$.
The $\text{vec}(\cdot)$ operator vectorizes tensors. The tensor entry $a_{i_1,i_2,\ldots,i_N}$ of $\mathcal{A} \in \R^{I_1 \times \cdots \times I_N}$ corresponds to the entry of $\text{vec}(\mathcal{A})$ with index $(i_N-1)I_{N-1}\cdots I_1 + (i_{N-1} -1) I_{N-2} \cdots I_1 + \ldots + (i_2-1)I_1 + i_1$. Vectors and matrices of all zeros and all ones are denoted by $\mathbf{0}$ and $\mathbf{1}$, respectively, with shape inferred from context. The outer and Kronecker products are denoted by $\op$ and $\kron$, respectively. The outer product of \( n \) vectors \( \mathbf{a}^{(1)}, \mathbf{a}^{(2)}, \ldots, \mathbf{a}^{(n)} \) is the \( n \)-th order tensor $\mathcal{T} = \mathbf{a}^{(1)} \op \mathbf{a}^{(2)} \op \cdots \op \mathbf{a}^{(n)},$ whose entries are defined as ${t}_{i_1 i_2 \ldots i_n} = a^{(1)}_{i_1} a^{(2)}_{i_2} \cdots a^{(n)}_{i_n}.$ The Kronecker product of $\mathbf{A} \in \R^{I \times K}$ and $\mathbf{B} \in \R^{J \times L}$ is defined as
\[
\mathbf{A} \kron \mathbf{B} \triangleq 
\begin{bmatrix}
    a_{1,1}\mathbf{B} & \cdots & a_{1,K}\mathbf{B} \\ \vdots & & \vdots \\
    a_{I,1}\mathbf{B} & \cdots & a_{I,K}\mathbf{B}
\end{bmatrix} \in \R^{IJ \times KL}. 
\]

\section{Sparse Approximation and Polynomials}
\label{sec:connection_sparse_poly}
An $s$-sparse vector has the property that any product of $s+1$ distinct entries is zero. For example, if $\mathbf{x}$ has at most $s = 2$ nonzeros, then $x_i x_j x_k = 0$ for all distinct $i, j, k$. There are $\binom{n}{s+1}$ such monomial equations corresponding to all combinations of $s+1$ indices. Sparse approximation now amounts to solving $\mathbf{Ax} = \mathbf{b}$ together with the sparsity-inducing monomial equations.

Before developing methods, we consider the geometric interpretation {for a toy example}. Consider
\[
\underbrace{\begin{bmatrix}
    1 & -1 & 1 \\ 0 & 1 & -1
\end{bmatrix}}_{\mathbf{A}} 
\underbrace{\begin{bmatrix}
    x_1 \\ x_2 \\ x_3
\end{bmatrix}}_{\mathbf{x}} = 
\underbrace{\begin{bmatrix}
    1 \\ -1
\end{bmatrix}}_{\mathbf{b}}
\]
with two 1-sparse solutions $\hat{\mathbf{x}}_{1} = [0~-1~0]^\top$ and $\hat{\mathbf{x}}_{2} = [0~0~1]^\top$. The sparsity condition $s=1$ gives three equations: $x_1 x_2 = 0$, $x_2 x_3 = 0$, and $x_1 x_3 = 0$, overlapped in blue in Fig.~\ref{fig:intersection}. Their intersection, shown in red, characterizes vectors with at most one nonzero entry. More generally, for any $n$ and sparsity level $s$, the solution set of the sparsity-inducing equations is a union of hyperplanes through the origin. This set contains all vectors $\mathbf{x}$ with sparsity at most $s$. Prior knowledge of the solution's sparsity is not strictly required, as this approach is also able to recover solutions that are sparser than the specified level $s$.

In our toy example, the sparsity-inducing equations define a set of vectors $\mathbf{x}$ with at most one nonzero entry, thus, the solution must lie on one of the red lines in Fig.~\ref{fig:sparse_via_poly}. The solution must also satisfy the linear constraint $\mathbf{A} \mathbf{x} = \mathbf{b}$, which geometrically means finding the intersection points of the red lines with the green line representing $\mathbf{A} \mathbf{x} = \mathbf{b}$, as shown in Fig.~\ref{fig:l0}. The intersection points are the sparse solutions. This already highlights a key advantage of the polynomial formulation: it encodes \textit{all} sparse solutions, i.e., the sparsest solution need not be unique.
\begin{figure}[!t]
  \centering
  \begin{subfigure}[t]{0.4\linewidth}
    \centering
    \includegraphics[width=\linewidth]{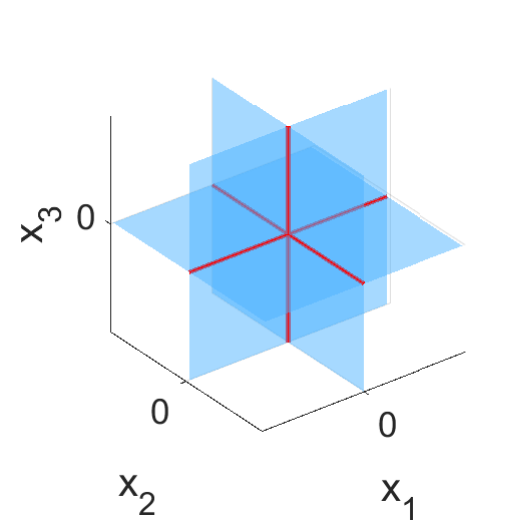}
    \caption{Sparsity-inducing equations and their intersection.}
    \label{fig:intersection}
  \end{subfigure}
  ~
  \begin{subfigure}[t]{0.4\linewidth}
    \centering
    \includegraphics[width=\linewidth]{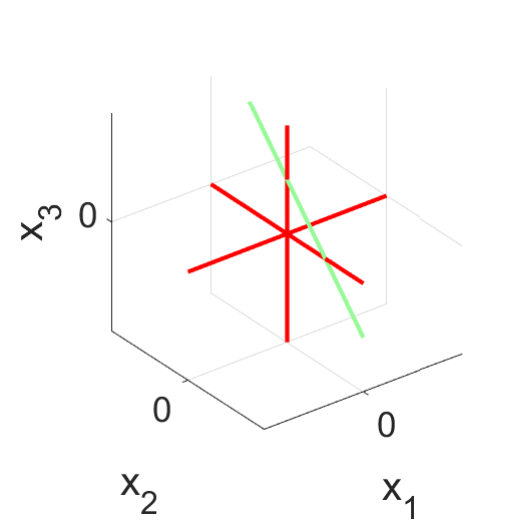}
    \caption{Sparse solutions as intersections with $\mathbf{A} \mathbf{x} = \mathbf{b}$.}
    \label{fig:l0}
  \end{subfigure}
  \caption{Solving the sparse approximation problem via systems of polynomial equations.}
  \label{fig:sparse_via_poly}
\end{figure}

\section{Tensors and Polynomial Equations}
\label{sec:tensors}

\subsection{Explicit Canonical Polyadic Decomposition}
\label{sec:rank&cpd}
A rank-1 tensor of order~$N$ is the outer product of $N$ nonzero vectors{; see Notation in Section~\ref{sec:intro}}. The rank of a tensor is the minimal number of rank-1 tensors required to represent it as a sum. This representation is called Canonical Polyadic Decomposition (CPD).
Remarkably, the CPD is unique (up to permutation and scaling) under mild conditions {\cite[Section~IV]{nikos_overview}}, unlike matrix rank decompositions, whose uniqueness typically requires conditions such as orthogonality.
\begin{definition}[Canonical Polyadic Decomposition] 
    \label{def:CPD}
    The CPD of a tensor \( \mathcal{A} \in \mathbb{R}^{I_1 \times I_2 \times \cdots \times I_K} \) is given by
    \[
    \mathcal{A} = \sum_{r=1}^{R} \mathbf{u}_r^{(1)} \op \mathbf{u}_r^{(2)}  \op \cdots  \op \mathbf{u}_r^{(K)} \triangleq \llbracket \mathbf{U}^{(1)}, \mathbf{U}^{(2)}, \ldots, \mathbf{U}^{(K)} \rrbracket,
    \]
    where \( R \) is the tensor rank, and \( \mathbf{u}_r^{(k)} \in \mathbb{R}^{I_k} \) are the factor vectors for each mode \( k \).
\end{definition}

Two main approaches for the computation of CPD exist: \emph{algebraic} methods and \emph{optimization-based} methods. Algebraic methods, such as those based on the generalized eigenvalue decomposition (GEVD)~{\cite{nikos_overview}}, are exact in the noiseless setting. In the presence of noise, the approximate solution of the algebraic algorithm can be used as an initial guess, which can be refined using optimization-based approaches. 
In this work, we opt to minimize a nonlinear least squares (NLS) objective to compute the CPD~\cite{nico_numerical}. Specifically, given a target rank~$R$, we fit the rank-$R$ CPD model
$
\llbracket \mathbf{U}^{(1)}, \mathbf{U}^{(2)}, \ldots, \mathbf{U}^{(K)} \rrbracket
$
with factor matrices $\mathbf{U}^{(1)}, \mathbf{U}^{(2)}, \ldots, \mathbf{U}^{(K)}$ to the data tensor~$\mathcal{T}$ by solving
\begin{equation}
\label{eq:explicit_CPD}
\min_{\mathbf{U}^{(1)}, \mathbf{U}^{(2)}, \ldots, \mathbf{U}^{(K)}} f = \lVert \llbracket \mathbf{U}^{(1)}, \mathbf{U}^{(2)}, \ldots, \mathbf{U}^{(K)} \rrbracket - \mathcal{T} \rVert_F^2.
\end{equation}

Gauss--Newton (GN) algorithm is a possible choice to minimize the NLS objective function. We employ an existing optimization framework built into Tensorlab~\cite{tensorlab3.0}. Only the key components need to be provided to the algorithm: the objective function, its gradient, and the Gramian of the Jacobian~\cite{nico_numerical}. When paired with trust-region strategies, GN is globally convergent under mild conditions ~\cite{numerical_opt}.

\subsection{Implicit CPD}
\label{sec:multiline_lscpd}
In the above, the tensor entries were explicitly given. However, sometimes the entries are only known implicitly. An example is the LS-CPD framework \cite{bousse_LS_CPD}, where the tensor appears as the solution to a linear system subject to a CPD constraint:
\begin{equation}
\label{eq:ls_cpd}
    \mathbf{Cy} = \mathbf{d}, \quad \mathbf{y} = \text{vec}(\mathcal{Y}), \quad \mathcal{Y} = \llbracket \mathbf{U}^{(1)}, \ldots, \mathbf{U}^{(N)} \rrbracket,
\end{equation}
where \( \mathbf{C} \in \mathbb{R}^{M \times K} \) and \( \mathbf{d} \in \mathbb{R}^{M} \). \( \mathcal{Y} \in \mathbb{R}^{I_1 \times \cdots \times I_N} \) is a rank-\( R \) tensor with factor matrices \( \mathbf{U}^{(n)} \in \mathbb{R}^{I_n \times R} \).
If the linear system is full rank, solving it yields a vectorized tensor of which the CPD can subsequently be computed. If the system is underdetermined, a CPD constrained solution can be obtained by minimizing
the NLS objective
\[
\min_{\mathbf{U}^{(1)}, \ldots, \mathbf{U}^{(N)}} f = \|\mathbf{C} \, \text{vec}( \llbracket \mathbf{U}^{(1)}, \ldots, \mathbf{U}^{(N)} \rrbracket ) -\mathbf{d}\|_F^2.
\]
The computation is analogous to the explicit CPD case and requires expressions for the objective function, its gradient and the Gramian of the Jacobian.

\subsection{Monomials and Rank-1 Tensors}
\label{sec:monomials and rank-1 tensors}
The sparsity-promoting monomials $x_{i_1} \cdots x_{i_{s+1}} = 0$ with $i_1 \neq \cdots \neq i_{s+1}$ appear in the symmetric rank-1 tensor
\[
\mathcal{X}
=
\underbrace{\mathbf{x} \op \cdots \op \mathbf{x}}_{s+1~\text{times}}.
\]
These monomials can be collected into a structured multilinear system of the form~\eqref{eq:ls_cpd} with rank $R=1$: 
\begin{equation}
    \label{eq:sparsity promoting LSCPD}
    \mathbf{Z} \,\text{vec}(\mathcal{X}) = \mathbf{0}.
\end{equation}
The matrix $\mathbf{Z}$ is defined as a selection matrix whose rows select the multilinear terms $x_{i_1} \cdots x_{i_{s+1}}$, with $i_1 \neq \cdots \neq i_{s+1}$ appearing in $\text{vec}(\mathcal{X}) = \kron_{k=1}^{s+1} \mathbf{x}$. 
Equation~\eqref{eq:sparsity promoting LSCPD} enforces all of these
selected terms to be equal to zero.

We extend this LS-CPD system with the data fitting term $\mathbf{Ax} = \mathbf{b}$, leading to the formulation
\begin{equation}
\label{eq:main_formulation}
    \begin{bmatrix}
        \mathbf{A} & \mathbf{0} \\ \mathbf{0} & \mathbf{Z} 
    \end{bmatrix}
     \begin{bmatrix}
         \mathbf{x} \\ \kron_{i=1}^{s+1} \mathbf{x}
     \end{bmatrix} = \begin{bmatrix}
         \mathbf{b} \\ \mathbf{0}
     \end{bmatrix}.
\end{equation}
The computation can follow the principles outlined above. In Section~\ref{sec:ls_cp_based}, we work out two efficient computational variants.

\section{Polynomial Equations and Multiple Solutions}
\label{sec:polyeq}
Formulation~\eqref{eq:main_formulation} describes a single solution $\mathbf{x}$ to a particular system of polynomial equations. We now show that all solutions of a polynomial system are described jointly by the null space of a matrix holding the polynomial coefficients.

\subsection{Roots as Structured Null Space Vectors}
\label{sec:roots_as_nullspace}

Consider a system of $p$ multivariate polynomials in $n$ complex variables
\begin{equation}
\label{eq:polysys}
   \begin{cases}
   f_1(x_1, \ldots, x_n) = 0 \\
   \makebox[\widthof{$f_1(x_1,\ldots,x_n)=0$}][c]{\vdots} \\
   f_p(x_1, \ldots, x_n) = 0
   \end{cases}
\end{equation}
with degrees $\deg(f_i) = d_i$ and maximal degree $d_0 = \max_i d_i$. The goal is to find the common roots $\mathbf{x}^{(q)} = [x^{(q)}_1, \ldots, x^{(q)}_n]^\top \in \mathbb{C}^n$, or, in the noisy case, the best fitting roots. We assume that~\eqref{eq:polysys} has $Q$ distinct roots.

For each root $\mathbf{x}^{(q)}$, the vector of all distinct monomials of total degree up to $d$ evaluated at this root is the multivariate Vandermonde vector
\begin{equation}
   \label{eq:vdm}
   \mathbf{v}_q(d) =
   \begin{bmatrix}
   1 & x^{(q)}_1 & x^{(q)}_2 & x^{(q)}_1 x^{(q)}_2 & \dots & x^{(q)d}_n
   \end{bmatrix}^\top,
\end{equation}
of length $N_d = \binom{n+d}{d}$. Every polynomial in~\eqref{eq:polysys} is linear in these monomials, so its coefficients can be collected in a matrix $\mathbf{M}_0 \in \mathbb{C}^{p \times N_{d_0}}$ such that
\begin{equation}
\label{eq:M0}
   \mathbf{M}_0 \mathbf{v}_q(d_0) = \mathbf{0},
   \qquad q = 1,\ldots,Q.
\end{equation}
Stacking all Vandermonde vectors into the multivariate Vandermonde matrix
$\mathbf{V}(d_0) = [\,\mathbf{v}_1(d_0) \; \cdots \; \mathbf{v}_{Q}(d_0)\,]$
gives
\begin{equation}
\label{eq:MV}
   \mathbf{M}_0 \mathbf{V}(d_0) = \mathbf{0}.
\end{equation}

With $d_0 = s+1$, equation~\eqref{eq:main_formulation} can be expressed in the form~\eqref{eq:M0} by dropping repeated monomials. Specifically, $\mathbf{b}$ has moved to the left-hand side, where it multiplies the leading entry $1$ of $\mathbf{v}_q(d_0)$, and $\kron_{i=1}^{d_0}\mathbf{x}$ is replaced by a vector $\mathbf{v}_q(d_0)$ holding a single copy of each monomial. The entries of $\mathbf{v}_q(d_0)$ are the distinct entries of the symmetric rank-1 tensor
\begin{equation}
\label{eq:rank1_root}
   \underbrace{\mathbf{x}^{(q)}_{h} \op \cdots \op \mathbf{x}^{(q)}_{h}}_{d_0~\text{times}},
   \quad \text{with} \quad
    \mathbf{x}^{(q)}_{h} = [1, x^{(q)}_1, \ldots, x^{(q)}_n]^\top.
\end{equation}
Thus, equation~\eqref{eq:MV} states that the null space of $\mathbf{M}_0$ contains $Q$ such rank-1 tensors, one per root. Solving a polynomial system is in this sense a subspace variant of~\eqref{eq:main_formulation}, and the task is to separate the $Q$ rank-1 terms that live in the null space.

\subsection{Recovering the Roots from the Null Space}
\label{sec:roots_from_nullspace}

Assume for now that the $Q$ Vandermonde vectors are linearly independent and that
\begin{equation}
\label{eq:nullity}
   \dim \ker (\mathbf{M}_0) = Q,
\end{equation}
so that the columns of $\mathbf{V}(d_0)$ form a basis of the null space of $\mathbf{M}_0$; the general case is treated in Section~\ref{sec:macaulay_method}. We write $\mathbf{V}(d)$ with $d = d_0$ in this subsection.

In practice $\mathbf{V}(d)$ is not available, since it depends on the unknown roots. However, it can be obtained from a numerical (orthonormal) basis $\mathbf{K}(d) = [\mathbf{k}_1~\cdots~\mathbf{k}_Q]$ for the null space of $\mathbf{M}_0$, obtained from its singular value decomposition. As $\mathbf{K}(d)$ and $\mathbf{V}(d)$ span the same subspace, they are related through an unknown nonsingular matrix $\mathbf{W} \in \mathbb{C}^{Q \times Q}$:
\begin{equation}
\label{eq:KVW}
   \mathbf{K}(d) = \mathbf{V}(d) \mathbf{W}^{\top}.
\end{equation}

Every basis vector $\mathbf{k}_j = \sum_{q} w_{jq} \mathbf{v}_q(d)$ is a mixture of the same $Q$ symmetric rank-1 tensors~\eqref{eq:rank1_root}, and only the coefficients $w_{jq}$ differ. Collecting these coefficients in an additional mode gives an explicit CPD whose rank-1 terms encode the roots. {Indeed, with $\mathbf{X}_h = [\,\mathbf{x}^{(1)}_h \; \cdots \; \mathbf{x}^{(Q)}_h\,]$, the tensor $\mathcal{K} = \llbracket \mathbf{X}_h, \; \ldots, \; \mathbf{X}_h, \; \mathbf{W} \rrbracket$ can be obtained from $\mathbf{K}(d)$ by reshaping. The CPD is formally guaranteed to return all Q roots} of~\eqref{eq:polysys} at once. To explain this, let us consider a further reshaping of $\mathcal{K}$ into a third-order tensor $\mathcal{Y} \in \mathbb{C}^{N_{d-1} \times (n+1) \times Q}$. It is easy to show that this tensor admits the CPD
\begin{equation}
\label{eq:cpd_null}
   \mathcal{Y} = \llbracket \mathbf{V}(d-1), \mathbf{V}(1), \mathbf{W}\rrbracket,
\end{equation}
whose second factor matrix contains the roots:
\begin{equation}
\label{eq:V1}
   \mathbf{V}(1) = \begin{bmatrix}
       1 & 1 & \cdots & 1 \\
       x_1^{(1)} & x_1^{(2)} & \cdots & x_1^{(Q)} \\
       \vdots & \vdots & & \vdots \\
       x_n^{(1)} & x_n^{(2)} & \cdots & x_n^{(Q)}
   \end{bmatrix}.
\end{equation}
Under~\eqref{eq:nullity}, $\mathbf{V}(d-1)$ and $\mathbf{W}$ have full column rank $Q$, and the columns of $\mathbf{V}(1)$ are pairwise non-collinear because the $Q$ roots are distinct. This guarantees the uniqueness of~\eqref{eq:cpd_null}; moreover, the factors can be computed by means of a simple matrix eigenvalue decomposition (EVD), see~\cite[Section IV]{nikos_overview} and references therein.

\subsection{The Macaulay Construction}
\label{sec:macaulay_method}
For a general polynomial system, assumption~\eqref{eq:nullity} fails. If $\dim \ker (\mathbf{M}_0) > Q$, the null space contains additional directions that do not correspond to any root, so the columns of $\mathbf{V}(d_0)$ no longer form a basis, and~\eqref{eq:KVW} does not hold.

The spurious directions can be removed by adding rows that every root still satisfies. Since $f_i(\mathbf{x}^{(q)}) = 0$ implies $\mathbf{x}^{(q)\boldsymbol{\alpha}} f_i(\mathbf{x}^{(q)}) = 0$ for every monomial $\mathbf{x}^{\boldsymbol{\alpha}}$, the shifted polynomials $\mathbf{x}^{\boldsymbol{\alpha}} f_i$ vanish at all $Q$ roots, so their coefficients give extra rows that constrain $\mathbf{v}_q(d)$ without excluding any root.

For a degree $d \geq d_0$, the Macaulay matrix $\mathbf{M}(d)$ collects, in its rows, the coefficients of all shifted polynomials $\mathbf{x}^{\boldsymbol{\alpha}} f_i$ with $\deg(\mathbf{x}^{\boldsymbol{\alpha}}) \in \{0,1,\ldots,d-d_i\}$, expressed in the monomials of $\mathbf{v}_q(d)$~\cite{macaulay_matrix}. Its rows include those of $\mathbf{M}_0$, obtained for $\boldsymbol{\alpha} = \mathbf{0}$, and every root still satisfies
\begin{equation}
\label{eq:macaulay_null}
   \mathbf{M}(d) \mathbf{v}_q(d) = \mathbf{0},
   \qquad
   \mathbf{M}(d) \mathbf{V}(d) = \mathbf{0}.
\end{equation}
The structured vectors of interest therefore remain in the null space at every degree, while the added rows keep removing directions that do not correspond to roots.

If the system has a finite number of roots, the dimension of the null space of $\mathbf{M}(d)$ stabilizes for $d \geq d^*$. The so-called degree of regularity $d^*$ is the lowest degree for which the null space is spanned only by vectors corresponding to the roots~{\cite{macaulay_matrix, macaulay_poly,Cox1998}}, that is,
\begin{equation}
\label{eq:regularity}
   \ker(\mathbf{M}(d)) = \operatorname{span}\{\mathbf{v}_1(d), \ldots, \mathbf{v}_Q(d)\},
   \qquad d \geq d^*.
\end{equation}
At the degree of regularity the situation of Section~\ref{sec:roots_from_nullspace} is restored, with $\mathbf{M}(d^*)$ in the role of $\mathbf{M}_0$, and the roots follow from a variant of~\eqref{eq:cpd_null} constructed for $d \geq d^*$. In practice neither $Q$ nor $d^*$ is known beforehand, so $d$ is increased until the dimension of the null space of $\mathbf{M}(d)$ stabilizes, and that dimension is taken as the number of roots.

For sparse approximation, the polynomial system consists of $\mathbf{Ax} = \mathbf{b}$ and the $\binom{n}{s+1}$ sparsity-inducing equations. The latter turn the infinitely many solutions of the underdetermined $\mathbf{Ax} = \mathbf{b}$ into a finite set: if every $s$ columns of $\mathbf{A}$ are linearly independent, each support of size at most $s$ carries at most one solution, so there are at most $\binom{n}{s}$ roots~\cite{FoucartRauhut2013}. Computing $\mathbf{K}(d)$ at the degree of regularity and decomposing~\eqref{eq:cpd_null} therefore returns \textit{all} of them at once in $\mathbf{V}(1)$. Solutions that are sparser than $s$ are included as well. The limitation is the computational cost. The Macaulay matrix has $N_d = \binom{n+d}{d}$ columns and a comparable number of rows, so {computing} $\mathbf{K}(d)$ may be expensive. As a speed-up, the columns of $\mathbf{M}(d)$ corresponding to monomials that need to be zero, can be dropped, and the null space of the resulting smaller matrix computed.

\section{Efficient Optimization-Based Algorithms}
\label{sec:ls_cp_based}
In LS-CPD problems, it is generally advantageous to exploit the structure of the system matrix, see \cite{hendrikx2022block} and references therein. This is also true in the present setting. The matrix $\mathbf{Z}$ in our core equation~\eqref{eq:main_formulation} is a large matrix of size $(m+n(n-1)\cdots(n-s)) \times (n+n^{s+1})$; however, as a selection matrix it is just a part of an identity matrix. Exploiting this fact, the GN quantities for the NLS optimization can be expressed in terms of elementary symmetric polynomials (ESPs). The latter can be evaluated cheaply, leading to computationally efficient implementations; see Section~\ref{sec:complexity}.

\begin{definition}[Elementary Symmetric Polynomial \cite{elem_sym_poly}]
\label{def:elem_sym_poly}
ESP of degree $d$ in $n$ variables is obtained by summing together all monomials of degree $d$ in $n$ variables:
\begin{equation*}
    e_d(z_1,\ldots,z_n)
    =
    \sum_{1 \leq i_1 < \cdots < i_d \leq n}
    z_{i_1}\cdots z_{i_d},
\end{equation*}
with $e_0(z_1,\ldots,z_n)=1$ and $e_d(z_1,\ldots,z_n)=0$ for $d>n$.
\end{definition}
In particular, we are interested in the following ESPs:
\begin{align}
\label{eq:ESP}
    E_{s+1}
    &:= e_{s+1}(x_1^2,\ldots,x_n^2),
    \\
\notag
    E_s^{(i)}
    &:= e_s(x_1^2,\ldots,x_{i-1}^2,x_{i+1}^2,\ldots,x_n^2),
    \\
\notag
    E_{s-1}^{(i,j)}
    &:= e_{s-1}(x_1^2,\ldots,x_{i-1}^2,x_{i+1}^2,\ldots,
    x_{j-1}^2,x_{j+1}^2,\ldots,x_n^2).
\end{align}

\subsection{Algorithms}

We propose two algorithms based on the formulation in~\eqref{eq:main_formulation}. The first, {SESP-D}, jointly optimizes data fitting and sparsity promotion. The second, SESP-P, seeks the sparsest solution that exactly fits the data. Both algorithms return a single solution for a given initialization; when multiple sparse solutions exist, different solutions may be obtained by reinitializing.

{\emph{Sparsification via ESP - Direct:}}
SESP-D applies the NLS approach to:
\begin{equation}
\label{eq:elscpd_formulation}
    \begin{bmatrix}
        \mathbf{A} & \mathbf{0} \\
        \mathbf{0} & \lambda \mathbf{Z}
    \end{bmatrix}
    \begin{bmatrix}
        \mathbf{x} \\
        \kron_{i=1}^{s+1} \mathbf{x}
    \end{bmatrix}
    =
    \begin{bmatrix}
        \mathbf{b} \\
        \mathbf{0}
    \end{bmatrix}.
\end{equation}
The scalar $\lambda$ is introduced to balance data fitting and sparsity {promotion.} By exploiting the structure of $\mathbf{Z}$ (Appendix~\ref{appx:structure_M}), the following expressions for the objective function, gradient, and Jacobian Gramian are derived in Appendix~\ref{appx:elscpd_derivation}:
\begin{align}
\label{eq:elscpd_obj_grad_gram}
    f =&~ \tfrac{1}{2}  \| \mathbf{Ax} - \mathbf{b} \|_2^2
    + \tfrac{1}{2} \lambda^2 (s+1)! E_{s+1} \\
\notag
    \mathbf{J}^\top \mathbf{r} =&~ \mathbf{A}^\top (\mathbf{Ax} - \mathbf{b})
    + \lambda^2 (s+1)! \sum_{i=1}^n E_s^{(i)} x_i \mathbf{e}_i \\
\notag
    \mathbf{J}^\top \mathbf{J} =& \mathbf{A}^\top \mathbf{A}
    + \lambda^2 (s+1)! \big( \sum_{i=1}^n E_s^{(i)} \mathbf{e}_i \mathbf{e}_i^\top \\ &
\notag
    + \sum_{\substack{i,j=1 \\ i \neq j}}^n E_{s-1}^{(i,j)} x_i x_j \mathbf{e}_i \mathbf{e}_j^\top \big),
\end{align}

{\emph{Sparsification via ESP - Projected:}}
The linear system $\mathbf{Ax}=\mathbf{b}$ is first solved to obtain a particular solution $\mathbf{x}_p$. The null space $\mathbf{V}$ of $\mathbf{A}$ is then computed, yielding the parametrization $\mathbf{x} = \mathbf{x}_p + \mathbf{V}\mathbf{y}$. The goal is to determine $\mathbf{y}$ such that the resulting solution $\mathbf{x}$ is sparse. Sparsity is enforced through the matrix $\mathbf{Z}$, leading to
\[
\mathbf{Z}\,\operatorname{vec}(\mathcal{X})=\mathbf{0},
\qquad
\operatorname{vec}(\mathcal{X})
=
\kron_{i=1}^{s+1}
(\mathbf{x}_p+\mathbf{V}\mathbf{y}).
\]
The resulting NLS problem is solved using GN algorithm.

The optimization variable is $\mathbf{y}$, with the solution recovered through $\mathbf{x}=\mathbf{x}_p+\mathbf{V}\mathbf{y}$. For notational convenience, however, the objective function, gradient, and Jacobian Gramian are expressed in terms of $\mathbf{x}$, which can be readily computed from $\mathbf{y}$ at each iteration. Appendix~\ref{appx:nullspace_derivation} derives the expressions
\begin{align}
\label{eq:nullspace_obj_grad_gram}
    f &= \tfrac{1}{2}  (s+1)!E_{s+1} \\
\notag
    \mathbf{J}^\top \mathbf{r} &= (s+1)! \mathbf{V}^\top\sum_{i=1}^n E_s^{(i)} x_i \mathbf{e}_i \\
\notag
    \mathbf{J}^\top \mathbf{J} &= (s+1)! \mathbf{V}^\top \big( \sum_{i=1}^n E_s^{(i)} \mathbf{e}_i \mathbf{e}_i^\top + \sum_{\substack{i,j=1 \\ i \neq j}}^n E_{s-1}^{(i,j)} x_i x_j  \mathbf{e}_i \mathbf{e}_j^\top \big) \mathbf{V}.
\end{align}

\subsection{Computational Complexity of the LS-CPD-Based Methods}
\label{sec:complexity}

ESPs appearing in~\eqref{eq:elscpd_obj_grad_gram} and~\eqref{eq:nullspace_obj_grad_gram} involve sums with up to $\binom{n}{s+1}$ terms, and their naive evaluation requires $\mathcal{O}(n^{s+1})$ operations. However, ESPs have a structure which we exploit to reduce the complexity. The proof of Lemma~\ref{lemma:esp_recursion} (Appendix~\ref{appx:esp_recursion}) establishes a recursion which, when applied, evaluates the objective functions of {SESP-D} and {SESP-P} in $\mathcal{O}(ns)$ time.
\begin{lemma}
\label{lemma:esp_recursion}
$e_d(z_1, \ldots, z_n)$ can be computed in $O(nd)$ time.
\end{lemma}

Applying the same recursion directly to evaluate $\mathbf{J}^\top \mathbf{r}$ and $\mathbf{J}^\top \mathbf{J}$ would require $\mathcal{O}(n^2 s)$ and $\mathcal{O}(n^3 s)$ operations, respectively. Lemma~\ref{lemma:leave_one_out} and Lemma~\ref{lemma:leave_two_out} reduce these complexities to $\mathcal{O}(ns)$ and $\mathcal{O}(n^2 s)$, respectively. The proofs in Appendices~\ref{appx:leave_one_out} and~\ref{appx:leave_two_out} are constructive.

\begin{lemma}
\label{lemma:leave_one_out}
Let $i=1,\dots, n$ and $s \leq n$. All $n$ values $E_s^{(i)}$ can be computed simultaneously in $O(ns)$ time.
\end{lemma}

\begin{lemma}
\label{lemma:leave_two_out}
Let $1 \leq i < j \leq n$ and $s \leq n$. All $\binom{n}{2}$ values $E_{s-1}^{(i,j)}$ can be computed simultaneously in $O(n^2 s)$ time.
\end{lemma}

Table~\ref{tab:complexities} reports per-iteration complexities of the GN algorithm, excluding iteration-independent quantities such as $\mathbf{V}$ and $\mathbf{A}^\top \mathbf{A}$, which are precomputed once. The remaining terms arise as follows: $ns$ and $n^2 s$ from evaluating the ESPs; $nm$ from computing the residual $\mathbf{A}\mathbf{x} - \mathbf{b}$ in $f$ and the product $\mathbf{A}^\top(\mathbf{A}\mathbf{x} - \mathbf{b})$ in $\mathbf{J}^\top \mathbf{r}$; $n(n-m)$ from the substitution $\mathbf{x} = \mathbf{x}_p + \mathbf{V}\mathbf{y}$ and the multiplication by $\mathbf{V}^\top$ in $\mathbf{J}^\top \mathbf{r}$; and $n^2(n-m)$ from the multiplications by $\mathbf{V}^\top$ and $\mathbf{V}$ in $\mathbf{J}^\top \mathbf{J}$.
\begin{table}[h!]
    \centering
    \caption{Per-iteration computational complexities of SESP-D and SESP-P}
    \begin{tabular}{|c|c|c|c|c|}
        \hline
        & Calls per Iteration & SESP-D &SESP-P\\
        \hline
        $f$ & $1 + it_{TR}$ &  $nm+ns$ & $n(n-m)+ns$\\
        \hline
        $\mathbf{J}^\top \mathbf{r}$ & 1 & $nm+ns$& $n(n-m)+ns$ \\
        \hline
        $\mathbf{J}^\top \mathbf{J}$ & 1 & $n^2s$ & $n^2(n-m) +  n^2s$\\
        \hline
    \end{tabular}
    \label{tab:complexities}
\end{table}

Since the dogleg trust-region GN method is used, the objective $f$ is evaluated $1 + it_{\mathrm{TR}}$ times per iteration, where $it_{\mathrm{TR}}$ denotes the number of trust-region sub-iterations. At each iteration, the descent step $\mathbf{p}$ is obtained by solving $\mathbf{J}^\top \mathbf{J}\, \mathbf{p} = \mathbf{J}^\top \mathbf{r}$ via the pseudoinverse,
\[
\mathbf{p} = (\mathbf{J}^\top \mathbf{J})^\dagger (\mathbf{J}^\top \mathbf{r}),
\]
which has computational complexity $\mathcal{O}(n^3)$~\cite{nla_trefethen}.

In Table~\ref{tab:complexities}, the terms $ns$ and $n^2 s$ correspond to enforcing sparsity, while $nm$, $n(n-m)$, and $n^2(n-m)$ correspond to enforcing $\mathbf{A}\mathbf{x} = \mathbf{b}$. Since $s \ll m$, the sparsity-enforcing terms do not dominate the overall complexity, confirming that polynomials can be used to design computationally efficient sparse approximation algorithms.

\section{Experiments}
\label{sec:comparison}
This section compares the proposed methods with BP/BPDN and OMP. Section~\ref{sec:implementation_algs} describes the MATLAB implementations and stopping criteria. Sections~\ref{sec:experiment_setup}, \ref{sec:initialization}, and~\ref{sec:regularization_param} detail the experimental setup, initialization strategies, and the choice of regularization parameter, respectively. Recovery accuracy and average runtime are evaluated for both the noiseless (Section~\ref{sec:noiseless_experiment}) and noisy (Section~\ref{sec:noisy_experiments}) settings. The algebraic framework allows explicit control over the target sparsity $s$ and enables the recovery of multiple {valid} sparse representations. This is illustrated in {Section~\ref{sec:sparse_modeling_experiment}. All }computations were performed in MATLAB~R2024b on a Lenovo ThinkPad~T480s equipped with an Intel\textsuperscript{\textregistered} Core\textsuperscript{TM}~i5-8350U CPU~@~1.70~GHz and 8~GB of RAM.

\subsection{MATLAB Implementation}
\label{sec:implementation_algs}
 
\subsubsection{SESP-D and SESP-P}
 
The Tensorlab~\cite{tensorlab3.0} implementation of the GN-NLS dogleg trust-region algorithm is used for both SESP-D and SESP-P. The function and step tolerances are $\varepsilon^2$ and $\varepsilon$, respectively, where $\varepsilon$ denotes machine precision. A maximum of $2000$ iterations is used.

At each iteration, the following stopping criterion based on support-restricted least squares is evaluated. Let $\mathbf{x}_k$ denote the current iterate, and let $S_k = \mathrm{supp}_s(\mathbf{x}_k)$ denote the index set of its $s$ largest-magnitude entries, where $s$ is the target sparsity supplied to the algorithm. Define the support-restricted least-squares solution $\mathbf{z} \in \mathbb{R}^n$ by
\begin{equation*}
    z_j = \begin{cases} \bigl(\mathbf{A}_{S_k}^\dagger \mathbf{b}\bigr)_j & j \in S_k, \\ 0 & j \notin S_k, \end{cases}
\end{equation*}
where $\mathbf{A}_{S_k} \in \mathbb{R}^{m \times s}$ is the submatrix of $\mathbf{A}$ formed by retaining only the columns indexed by $S_k$. The algorithm terminates when $\|\mathbf{A}\mathbf{z} - \mathbf{b}\|_2 \leq \tau$, where $\tau = 10^{-10}$ in the noiseless case and $\tau = \sigma\sqrt{m}$ in the noisy case. The latter threshold matches the expectation $\mathbb{E}[\|\mathbf{n}\|_2] = \sigma\sqrt{m}$ for additive Gaussian noise $\mathbf{n} \sim \mathcal{N}(\mathbf{0}, \sigma^2 \mathbf{I}_m)$. The algorithms return $\hat{\mathbf{x}} = \mathbf{z}$ as the computed solution to the sparse approximation problem.
 
\subsubsection{BP/BPDN}
 
In the noiseless case, the BP problem
\begin{equation*}
    \min_{\mathbf{x}} \|\mathbf{x}\|_1 \quad \text{subject to} \quad \mathbf{Ax} = \mathbf{b}
\end{equation*}
is solved following~\cite{basis_pursuit} using MATLAB's \texttt{linprog} interior-point solver with optimality tolerance $10^{-8}$, constraint tolerance $10^{-6}$, and a maximum of $2000$ iterations. In the noisy setting, where $\sigma$ denotes the standard deviation of the additive noise, the BPDN problem
\begin{equation*}
    \min_{\mathbf{x}} \|\mathbf{x}\|_1 \quad \text{subject to} \quad
    \|\mathbf{Ax} - \mathbf{b}\|_2 \leq \sigma\sqrt{m}
\end{equation*}
is solved using SPGL1~\cite{spgl1,spgl2}.

\subsubsection{OMP}
OMP~\cite{OMP_Tropp_Gilbert} greedily selects at each iteration the column of $\mathbf{A}$ most correlated with the current residual, incrementally building a sparse approximation of $\mathbf{b}$. The algorithm terminates when the residual norm falls below a prescribed threshold: $\|\mathbf{Ax} - \mathbf{b}\|_2 \leq 10^{-10}$ in the noiseless case, and $\|\mathbf{Ax} - \mathbf{b}\|_2 \leq \sigma\sqrt{m}$ in the noisy case.

\subsection{General Experimental Setup}
\label{sec:experiment_setup}

The matrix $\mathbf{A} \in \mathbb{R}^{m \times n}$ has i.i.d.\ entries drawn from the standard normal distribution, with columns subsequently normalized to unit $\ell_2$ norm. Unless otherwise specified, $m=32$ and $n=64$. The vector $\mathbf{x}_{\mathrm{true}} \in \mathbb{R}^{n}$ has $s$ nonzero entries, with support chosen uniformly at random and nonzero values drawn independently from the standard normal distribution.

In the noiseless setting, the measurement vector $\mathbf{b}$ is computed as $\mathbf{b} = \mathbf{A}\mathbf{x}_{\mathrm{true}}$. A computed solution $\hat{\mathbf{x}}$ is declared successful if $\|\hat{\mathbf{x}} - \mathbf{x}_{\mathrm{true}}\|_2 / \|\mathbf{x}_{\mathrm{true}}\|_2 < 10^{-5}$. The success rate is defined as the ratio of successful recoveries over $100$ trials.

In the noisy setting, the measurement vector is given by $\mathbf{b} = \mathbf{A}\mathbf{x}_{\mathrm{true}} + \mathbf{n}$, with additive Gaussian noise $\mathbf{n} \sim \mathcal{N}(\mathbf{0}, \sigma^2 \mathbf{I})$. The signal-to-noise ratio is defined as
\begin{equation*}
    \mathrm{SNR} = 10\log_{10}\!\left(\frac{\|\mathbf{A}\mathbf{x}_{\mathrm{true}}\|_2^2 / m}{\sigma^2}\right).
\end{equation*} 
The relative error $\|\hat{\mathbf{x}} - \mathbf{x}_{\mathrm{true}}\|_2 / \|\mathbf{x}_{\mathrm{true}}\|_2$ is reported instead of the success rate. All algorithms are evaluated on the same randomly generated problem instance, i.e., the same $\mathbf{A},$ $\mathbf{x}_{\mathrm{true}}$, and $\mathbf{b}$ within each trial. The sparsity level $s$ is assumed known and provided to the algorithms. For SESP-D, $\lambda$ is fixed at $10^{-4}$; Section~\ref{sec:regularization_param} motivates this choice. For SESP-P, we use the particular solution $\mathbf{x}_p$ satisfying $\mathbf{A}\mathbf{x}_p=\mathbf{b}$ with noiseless or noisy measurements $\mathbf{b}$, depending on the setting.

\subsection{Initialization}
\label{sec:initialization}

SESP-D and SESP-P solve nonconvex optimization problems, and the GN solver is a local method. The choice of initialization may therefore affect the outcome. We briefly discuss the initialization strategy used for each method.

\subsubsection{SESP-D}

Let $\mathbf{x}_0 = \alpha \mathbf{z}$, where $\mathbf{z} \sim \mathcal{N}(\mathbf{0}, \mathbf{I}_n)$ and $\alpha > 0$ is a scaling parameter to be determined. To ensure that the data-fitting term $f_1 = \tfrac{1}{2}\|\mathbf{A}\mathbf{x} - \mathbf{b}\|_2^2$ and the sparsity-promoting term $f_2 = \tfrac{1}{2}(s+1)!\,\lambda\, E_{s+1}(\mathbf{x}^{\cdot 2})$ in~\eqref{eq:elscpd_obj_grad_gram} contribute equally at initialization, we choose $\alpha$ such that their expected values are equal. A direct calculation yields
\begin{equation*}
    \mathbb{E}\left[f_1\right] = \frac{\alpha^2 n + s}{2}, \qquad
    \mathbb{E}\left[f_2\right] = \alpha^{2(s+1)} \prod_{j=0}^{s}(n - j),
\end{equation*}
where the expectations are taken over $\mathbf{z}$. Solving $\mathbb{E}[f_1] = \mathbb{E}[f_2]$ for $\alpha$ yields the initialization scale. Fig.~\ref{fig:ExperimentInitialization} shows the effect of increasing the number of random restarts on success rate: a single initialization suffices at low sparsity, while for higher sparsity reinitializations may improve the result.

\begin{figure}[ht]
  \centering
  \begin{tikzpicture}
    \begin{axis}[
      width=1\columnwidth,
      height=0.36\columnwidth,
      xlabel={Sparsity $s$},
      ylabel={Success Rate},
      xmin=1, xmax=15,
      ymin=0.25, ymax=1,
      ytick={0.25, 0.5, 0.75, 1},
      yticklabels={0.25, 0.5, 0.75, 1},
      xtick={1,3,5,7,9,11,13,15},
      grid=both,
      grid style={line width=0.4pt, draw=gray!30},
      major grid style={line width=0.6pt, draw=gray!50},
      tick label style={font=\small},
      label style={font=\small},
      legend style={
        at={(0,0)},
        anchor=south west,
        legend columns=3,
        font=\scriptsize,
        cells={anchor=west},
        row sep=-2pt,
        nodes={inner sep=1.5pt},
        /tikz/every even column/.append style={column sep=0.4em},
      },
      cycle list={
        {color={rgb,1:red,0.302;green,0.686;blue,0.290}, mark=+,        line width=1pt, mark size=2pt},
        {color={rgb,1:red,0.902;green,0.671;blue,0.008}, mark=triangle, line width=1pt, mark size=2pt, mark options={rotate=90}},
        {color={rgb,1:red,0.851;green,0.373;blue,0.008}, mark=x,        line width=1pt, mark size=2pt},
        {color={rgb,1:red,0.459;green,0.439;blue,0.702}, mark=star,     line width=1pt, mark size=2pt},
        {color={rgb,1:red,0.106;green,0.620;blue,0.467}, mark=pentagon, line width=1pt, mark size=2pt},
      },
    ]

    \addplot table[x=s, y=init1]{figureInitialization.dat};
    \addlegendentry{1 init}

    \addplot table[x=s, y=init2]{figureInitialization.dat};
    \addlegendentry{2 init}

    \addplot table[x=s, y=init5]{figureInitialization.dat};
    \addlegendentry{5 init}

    \addplot table[x=s, y=init10]{figureInitialization.dat};
    \addlegendentry{10 init}

    \addplot table[x=s, y=init20]{figureInitialization.dat};
    \addlegendentry{20 init}

    \end{axis}
  \end{tikzpicture}
  \caption{Success rate versus sparsity for SESP-D with different
           numbers of random initializations.}
  \label{fig:ExperimentInitialization}
\end{figure}
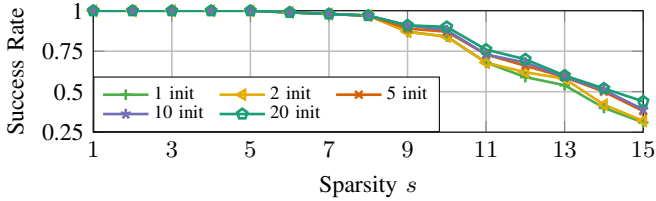

\subsubsection{SESP-P}

The SESP-P algorithm minimizes $E_{s+1}(\mathbf{x}^{\cdot 2})$ over the affine solution set via GN iterations on the null-space coordinate $\mathbf{y}$, where $\mathbf{x} = \mathbf{x}_p + \mathbf{V}\mathbf{y}$. Because the objective $E_{s+1}(\mathbf{x}_0^{\cdot 2})$ increases with the magnitude of each entry of $\mathbf{x}_0$, we initialize by minimizing $\|\mathbf{x}_p + \mathbf{V}\mathbf{y}_0\|_2$, which gives $\mathbf{y}_0 = -\mathbf{V}^\dagger \mathbf{x}_p$. Since the matrix $\mathbf{V}$ contains an orthonormal basis for the null space of $\mathbf{A}$, this simplifies to $\mathbf{y}_0 = -\mathbf{V}^\top \mathbf{x}_p$.

\subsubsection{Warm Start from BP/BPDN}

The BP/BPDN solution $\mathbf{x}_{\mathrm{BP/BPDN}}$ can be used to warm-start SESP-D and SESP-P. For SESP-D, we set $\mathbf{x}_0 = \mathbf{x}_{\mathrm{BP/BPDN}}$. For SESP-P, we set $\mathbf{y}_0 = \mathbf{V}^\top (\mathbf{x}_{\mathrm{BP/BPDN}} - \mathbf{x}_p)$. This strategy provides a starting point that is optimal in the $\ell_1$ sense.

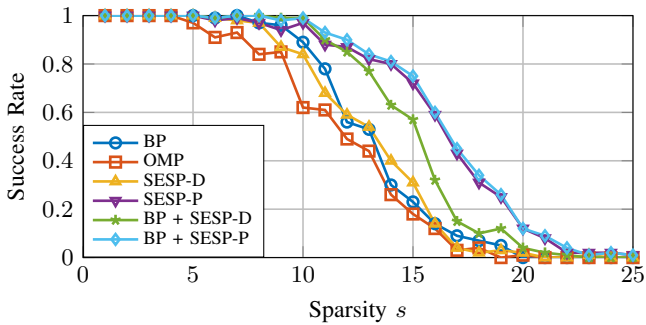
\begin{figure}[ht]
  \centering
  \begin{tikzpicture}
    \begin{axis}[
      width=\columnwidth,
      height=0.54\columnwidth,
      xlabel={Sparsity $s$},
      ylabel={Success Rate},
      xmin=0, xmax=25,
      ymin=0, ymax=1,
      ytick={0, 0.2, 0.4, 0.6, 0.8, 1.0},
      xtick={0, 5, 10, 15, 20, 25},
      grid=both,
      grid style={line width=0.4pt, draw=gray!30},
      major grid style={line width=0.6pt, draw=gray!50},
      tick label style={font=\small},
      label style={font=\small},
      legend style={
        at={(0,0)},
        anchor=south west,
        font=\scriptsize,
        cells={anchor=west},
        legend columns=1,
        row sep=-2pt,
        nodes={inner sep=1.5pt},
      },
      cycle list={
        {color={rgb,1:red,0.000;green,0.447;blue,0.741}, mark=o,        line width=1pt, mark size=2pt},
        {color={rgb,1:red,0.851;green,0.325;blue,0.098}, mark=square,   line width=1pt, mark size=2pt},
        {color={rgb,1:red,0.929;green,0.694;blue,0.125}, mark=triangle, line width=1pt, mark size=2pt},
        {color={rgb,1:red,0.494;green,0.184;blue,0.557}, mark=triangle, line width=1pt, mark size=2pt, mark options={rotate=180}},
        {color={rgb,1:red,0.466;green,0.675;blue,0.188}, mark=star, line width=1pt, mark size=2pt},
        {color={rgb,1:red,0.302;green,0.745;blue,0.933}, mark=diamond, line width=1pt, mark size=2pt},
      },
    ]

    \addplot table[x=s, y=bp]{figureSuccessRate.dat};
    \addlegendentry{BP}

    \addplot table[x=s, y=omp]{figureSuccessRate.dat};
    \addlegendentry{OMP}

    \addplot table[x=s, y=sespd]{figureSuccessRate.dat};
    \addlegendentry{SESP-D}

    \addplot table[x=s, y=sespp]{figureSuccessRate.dat};
    \addlegendentry{SESP-P}

    \addplot table[x=s, y=bp_sespd]{figureSuccessRate.dat};
    \addlegendentry{BP + SESP-D}

    \addplot table[x=s, y=bp_sespp]{figureSuccessRate.dat};
    \addlegendentry{BP + SESP-P}

    \end{axis}
  \end{tikzpicture}
  \caption{Success rate versus sparsity over 100 trials in the noiseless setting. SESP-D uses 1 random initialization.}
  \label{fig:ExperimentSuccess}
\end{figure}

The recovery performance for different initialization strategies in the noiseless setting is reported in Fig.~\ref{fig:ExperimentSuccess}, alongside BP and OMP for reference. All methods achieve perfect recovery at low sparsity; as sparsity increases, SESP-D performs similarly to BP, while SESP-P outperforms it, showing that the polynomial-based formulation can surpass $\ell_1$-based recovery on its own. Warm-starting SESP-D and SESP-P from the BP solution (BP~+~SESP-D and BP~+~SESP-P) further improves performance beyond BP alone. Since BP finds the $\ell_1$ optimal solution to $\mathbf{Ax} = \mathbf{b}$ instead of the $\ell_0$, our polynomial-based refinement further sparsifies that solution.

\subsection{Regularization Parameter}
\label{sec:regularization_param}

The SESP-D formulation~\eqref{eq:elscpd_formulation} includes the regularization parameter $\lambda$, which controls the relative weight of the sparsity-promoting term in the objective.

In the noiseless setting, the success rate is insensitive to the choice of $\lambda$ over several orders of magnitude for low-to-medium sparsity. At higher sparsity, the choice of $\lambda$ is not very critical either, but quite small values of $\lambda$ tend to be preferable. They place relatively more weight on the data-fitting term compared to the sparsity-promoting term. Fig.~\ref{fig:els_noiseless} shows the success rate as a function of both $s$ and $\lambda$; for each $(\lambda, s)$ pair, $5$ random initializations are used across $20$ trials, and a trial is declared successful if at least one initialization succeeds. Yellow regions indicate successful recovery, while dark blue regions indicate failure across all $5$ initializations.

In the noisy setting, $\lambda$ has a similar effect. Fig.~\ref{fig:els_noisy20db} shows the median error over $20$ trials as a function of $s$ and $\lambda$, for additive Gaussian noise with SNR~$=20$~dB. A recovery with relative error below $10^{-1}$ is considered successful at this SNR.

Based on Fig.~\ref{fig:heatmap}, $\lambda = 10^{-4}$ is used for all experiments.

\begin{figure}[h!]
\centering
\begin{subfigure}{0.48\columnwidth}
\centering
\begin{tikzpicture}
\begin{axis}[
    width=\columnwidth,
    height=\columnwidth,
    view={0}{90},
    xlabel={Sparsity $s$},
    ylabel={$\lambda$},
    xtick={1,3,6,9,12,15},
    xticklabels={1,3,6,9,12,15},
    ytick={1,5,9,13,17},
    yticklabels={$10^{-8}$,$10^{-4}$,$10^{0}$,$10^{4}$,$10^{8}$},
    tick label style={font=\small},
    label style={font=\small},
    colormap/viridis,
    colorbar horizontal,
    colorbar style={
        at={(0.5,1.45)},
        anchor=north,
        width=0.85*\pgfkeysvalueof{/pgfplots/parent axis width},
        height=0.05*\pgfkeysvalueof{/pgfplots/parent axis width},
        xlabel={Success rate},
        label style={font=\small},
        xtick={0,0.5,1.0},
    },
    enlarge x limits=false,
    enlarge y limits=false,
]
\addplot[
    matrix plot*,
    point meta=explicit,
    mesh/cols=15,
] table[x=x, y=y, meta=z] {figureHeatmapNoiseless.dat};
\end{axis}
\end{tikzpicture}
\caption{Noiseless setting.}
\label{fig:els_noiseless}
\end{subfigure}
\hfill
\begin{subfigure}{0.48\columnwidth}
\centering
\begin{tikzpicture}
\begin{axis}[
    width=\columnwidth,
    height=\columnwidth,
    view={0}{90},
    xlabel={Sparsity $s$},
    ylabel={},
    xtick={1,3,6,9,12,15},
    xticklabels={1,3,6,9,12,15},
    ytick={1,5,9,13,17},
    yticklabels={,,,,},
    tick label style={font=\small},
    label style={font=\small},
    colormap/viridis,
    colorbar horizontal,
    colorbar style={
        at={(0.5,1.45)},
        anchor=north,
        width=0.85*\pgfkeysvalueof{/pgfplots/parent axis width},
        height=0.05*\pgfkeysvalueof{/pgfplots/parent axis width},
        xlabel={Med. rel. error},
        label style={font=\small},
        xtick={-0.4,-0.1},
        xticklabels={$0.4$,$0.1$},
    },
    enlarge x limits=false,
    enlarge y limits=false,
]
\addplot[
    matrix plot*,
    point meta=explicit,
    mesh/cols=15,
] table[x=x, y=y, meta=z] {figureHeatmapNoisy.dat};
\end{axis}
\end{tikzpicture}
\caption{Noisy setting, SNR $= 20$\,dB.}
\label{fig:els_noisy20db}
\end{subfigure}
\caption{Recovery performance of SESP-D as a function of sparsity level $s$ and regularization parameter $\lambda$.}
\label{fig:heatmap}
\end{figure}
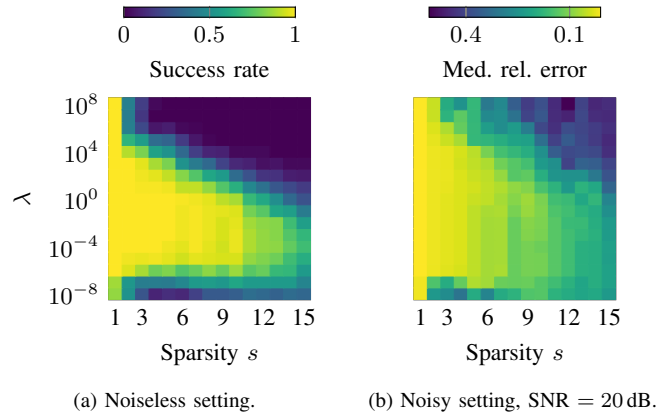

\subsection{Recovery and Runtime in the Noiseless Case}
\label{sec:noiseless_experiment}
Recovery performance of warm-started SESP-D and SESP-P in the noiseless case is compared to BP and OMP in Figs.~\ref{fig:ExperimentTroppFigure1} and~\ref{fig:ExperimentTroppFigure2}. Fig.~\ref{fig:ExperimentTroppFigure1} displays the success rate as a function of the number of measurements $m$ for a fixed sparsity level $s$ over 100 trials. SESP-D and SESP-P require fewer measurements than BP and OMP to successfully recover an $s$-sparse signal.

\definecolor{mycolor1}{rgb}{0.000,0.447,0.741}
\definecolor{mycolor2}{rgb}{0.851,0.325,0.098}
\definecolor{mycolor3}{rgb}{0.929,0.694,0.125}

\begin{figure}[ht]
  \centering
  \begin{tikzpicture}
    \begin{axis}[
      width=1\columnwidth,
      height=0.5\columnwidth,
      xlabel={Number of measurements $m$},
      ylabel={Success Rate},
      xmin=0, xmax=64,
      ymin=0, ymax=1,
      grid=both,
      grid style={line width=0.4pt, draw=gray!30},
      major grid style={line width=0.6pt, draw=gray!50},
      tick label style={font=\small},
      label style={font=\small},
      cycle list={
        {rgb,1:red,0.000;green,0.447;blue,0.741},
        {rgb,1:red,0.851;green,0.325;blue,0.098},
        {rgb,1:red,0.929;green,0.694;blue,0.125},
      },
      legend columns=4,
      legend style={
        at={(0.5,1.04)}, anchor=south,
        draw=none,
        font=\footnotesize,
        inner sep=1pt,
        /tikz/every even column/.append style={column sep=0.22cm},
      },
      legend cell align=left,
      legend image code/.code={%
        \draw[mark repeat=2, mark phase=2, #1]
          plot coordinates {(0cm,0cm) (0.2cm,0cm) (0.4cm,0cm)};%
      },
    ]

    \addplot[color=mycolor1, line width=1pt, forget plot] table[x=m,y=bp_s1]{figureNoiselessRecoveryVsMeasurements.dat};
    \addplot[color=mycolor1, dotted, line width=1pt, forget plot] table[x=m,y=omp_s1]{figureNoiselessRecoveryVsMeasurements.dat};
    \addplot[color=mycolor1, solid, mark=star, mark size=2pt, line width=1pt, forget plot, mark options={fill=white}] table[x=m,y=els_s1]{figureNoiselessRecoveryVsMeasurements.dat};
    \addplot[color=mycolor1, solid, mark=diamond, mark size=2pt, line width=1pt, forget plot, mark options={fill=white}] table[x=m,y=null_s1]{figureNoiselessRecoveryVsMeasurements.dat};

    \addplot[color=mycolor2, line width=1pt, forget plot] table[x=m,y=bp_s2]{figureNoiselessRecoveryVsMeasurements.dat};
    \addplot[color=mycolor2, dotted, line width=1pt, forget plot] table[x=m,y=omp_s2]{figureNoiselessRecoveryVsMeasurements.dat};
    \addplot[color=mycolor2, solid, mark=star, mark size=2pt, line width=1pt, mark options={fill=white}, forget plot] table[x=m,y=els_s2]{figureNoiselessRecoveryVsMeasurements.dat};
    \addplot[color=mycolor2, solid, mark=diamond, mark size=2pt, line width=1pt, mark options={fill=white}, forget plot] table[x=m,y=null_s2]{figureNoiselessRecoveryVsMeasurements.dat};

    \addplot[color=mycolor3, line width=1pt, forget plot] table[x=m,y=bp_s3]{figureNoiselessRecoveryVsMeasurements.dat};
    \addplot[color=mycolor3, dotted, line width=1pt, forget plot] table[x=m,y=omp_s3]{figureNoiselessRecoveryVsMeasurements.dat};
    \addplot[color=mycolor3, solid, mark=star, mark size=2pt, line width=1pt, mark options={fill=white}, forget plot] table[x=m,y=els_s3]{figureNoiselessRecoveryVsMeasurements.dat};
    \addplot[color=mycolor3, solid, mark=diamond, mark size=2pt, line width=1pt, mark options={fill=white}, forget plot] table[x=m,y=null_s3]{figureNoiselessRecoveryVsMeasurements.dat};

    \addlegendimage{black, line width=1pt}
    \addlegendentry{BP}
    \addlegendimage{black, dotted, line width=1pt}
    \addlegendentry{OMP}
    \addlegendimage{black, solid, mark=star, mark size=2pt, mark options={fill=white}, line width=1pt}
    \addlegendentry{BP + SESP-D}
    \addlegendimage{black, solid, mark=diamond, mark size=2pt, mark options={fill=white}, line width=1pt}
    \addlegendentry{BP + SESP-P}

    \end{axis}

    \begin{axis}[
      at={(rel axis cs:1,0.37)}, anchor=south east,
      hide axis,
      xmin=0, xmax=1, ymin=0, ymax=1,
      width=2.2cm, height=1.6cm,
      legend style={
        draw=black, fill=white,
        font=\tiny, cells={anchor=west},
        legend columns=1,
      },
    ]
    \addlegendimage{only marks, mark=square*, mark size=1pt, color=mycolor1}
    \addlegendentry{$s=4$}
    \addlegendimage{only marks, mark=square*, mark size=1pt, color=mycolor2}
    \addlegendentry{$s=12$}
    \addlegendimage{only marks, mark=square*, mark size=1pt, color=mycolor3}
    \addlegendentry{$s=20$}
    \end{axis}

  \end{tikzpicture}
  \caption{Effect of measurement count on recovery.}
  \label{fig:ExperimentTroppFigure1}
\end{figure}
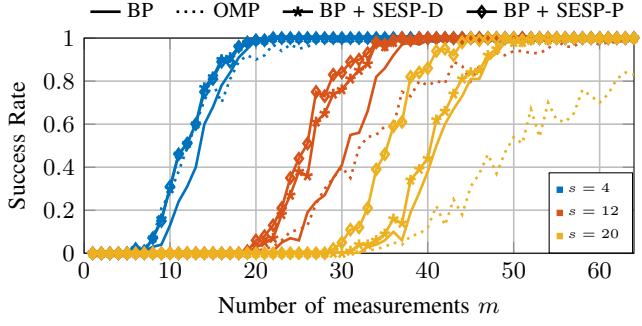

Fig.~\ref{fig:ExperimentTroppFigure2} displays the success rate as a function of the sparsity level $s$ for a fixed number of measurements $m$ over 100 trials. With $m = 32$ measurements, SESP-D and SESP-P recover a $12$-sparse signal in approximately $90\%$ of trials, compared to roughly $50\%$ for BP and OMP.

\definecolor{mycolor1}{rgb}{0.000,0.447,0.741}
\definecolor{mycolor2}{rgb}{0.851,0.325,0.098}
\definecolor{mycolor3}{rgb}{0.929,0.694,0.125}

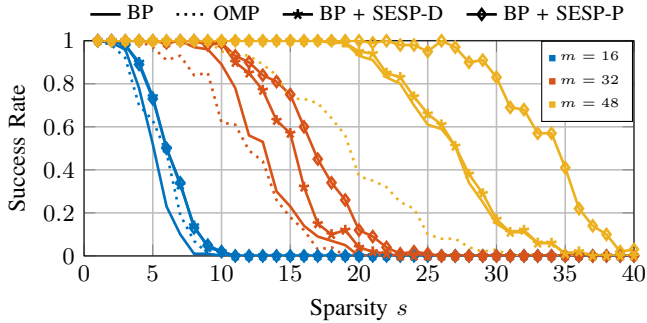
\begin{figure}[ht]
  \centering
  \begin{tikzpicture}
    \begin{axis}[
      width=1\columnwidth,
      height=0.5\columnwidth,
      xlabel={Sparsity $s$},
      ylabel={Success Rate},
      xmin=0, xmax=40,
      ymin=0, ymax=1,
      xtick={0, 5, 10, 15, 20, 25, 30, 35, 40},
      grid=both,
      grid style={line width=0.4pt, draw=gray!30},
      major grid style={line width=0.6pt, draw=gray!50},
      tick label style={font=\small},
      label style={font=\small},
      cycle list={
        {rgb,1:red,0.000;green,0.447;blue,0.741},
        {rgb,1:red,0.851;green,0.325;blue,0.098},
        {rgb,1:red,0.929;green,0.694;blue,0.125},
      },
      legend columns=4,
      legend style={
        at={(0.5,1.04)}, anchor=south,
        draw=none,
        font=\footnotesize,
        inner sep=1pt,
        /tikz/every even column/.append style={column sep=0.22cm},
      },
      legend cell align=left,
      legend image code/.code={%
        \draw[mark repeat=2, mark phase=2, #1]
          plot coordinates {(0cm,0cm) (0.2cm,0cm) (0.4cm,0cm)};%
      },
    ]

    \addplot[color=mycolor1, line width=1pt, forget plot] table[x=k,y=bp_m1]{figureNoiselessRecoveryVsSparsity.dat};
    \addplot[color=mycolor1, dotted, line width=1pt, forget plot] table[x=k,y=omp_m1]{figureNoiselessRecoveryVsSparsity.dat};
    \addplot[color=mycolor1, solid, mark=star, mark size=2pt, line width=1pt, forget plot, mark options={fill=white}] table[x=k,y=els_m1]{figureNoiselessRecoveryVsSparsity.dat};
    \addplot[color=mycolor1, solid, mark=diamond, mark size=2pt, line width=1pt, forget plot, mark options={fill=white}] table[x=k,y=null_m1]{figureNoiselessRecoveryVsSparsity.dat};

    \addplot[color=mycolor2, line width=1pt, forget plot] table[x=k,y=bp_m2]{figureNoiselessRecoveryVsSparsity.dat};
    \addplot[color=mycolor2, dotted, line width=1pt, forget plot] table[x=k,y=omp_m2]{figureNoiselessRecoveryVsSparsity.dat};
    \addplot[color=mycolor2, solid, mark=star, mark size=2pt, line width=1pt, mark options={fill=white}, forget plot] table[x=k,y=els_m2]{figureNoiselessRecoveryVsSparsity.dat};
    \addplot[color=mycolor2, solid, mark=diamond, mark size=2pt, line width=1pt, mark options={fill=white}, forget plot] table[x=k,y=null_m2]{figureNoiselessRecoveryVsSparsity.dat};

    \addplot[color=mycolor3, line width=1pt, forget plot] table[x=k,y=bp_m3]{figureNoiselessRecoveryVsSparsity.dat};
    \addplot[color=mycolor3, dotted, line width=1pt, forget plot] table[x=k,y=omp_m3]{figureNoiselessRecoveryVsSparsity.dat};
    \addplot[color=mycolor3, solid, mark=star, mark size=2pt, line width=1pt, mark options={fill=white}, forget plot] table[x=k,y=els_m3]{figureNoiselessRecoveryVsSparsity.dat};
    \addplot[color=mycolor3, solid, mark=diamond, mark size=2pt, line width=1pt, mark options={fill=white}, forget plot] table[x=k,y=null_m3]{figureNoiselessRecoveryVsSparsity.dat};

    \addlegendimage{black, line width=1pt}
    \addlegendentry{BP}
    \addlegendimage{black, dotted, line width=1pt}
    \addlegendentry{OMP}
    \addlegendimage{black, solid, mark=star, mark size=2pt, mark options={fill=white}, line width=1pt}
    \addlegendentry{BP + SESP-D}
    \addlegendimage{black, solid, mark=diamond, mark size=2pt, mark options={fill=white}, line width=1pt}
    \addlegendentry{BP + SESP-P}

    \end{axis}

    \begin{axis}[
      at={(rel axis cs:1,1)}, anchor=north east,
      hide axis,
      xmin=0, xmax=1, ymin=0, ymax=1,
      width=2.2cm, height=1.6cm,
      legend style={
        draw=black, fill=white,
        font=\tiny, cells={anchor=west},
        legend columns=1,
      },
    ]
    \addlegendimage{only marks, mark=square*, mark size=1pt, color=mycolor1}
    \addlegendentry{$m=16$}
    \addlegendimage{only marks, mark=square*, mark size=1pt, color=mycolor2}
    \addlegendentry{$m=32$}
    \addlegendimage{only marks, mark=square*, mark size=1pt, color=mycolor3}
    \addlegendentry{$m=48$}
    \end{axis}

  \end{tikzpicture}
  \caption{Recovery with a fixed measurement count.}
  \label{fig:ExperimentTroppFigure2}
\end{figure}

Average runtime as a function of matrix size $n$ or sparsity level $s$ is reported in Fig.~\ref{fig:runtime_comparison}, with all timings measured using MATLAB timing functions and averaged over 20 trials. Since BP achieves perfect recovery for all $n$ and $s$ considered, warm-starting SESP-D and SESP-P from the BP solution would be redundant here and is omitted.

\definecolor{bpCol}{RGB}{0,114,189}
\definecolor{ompCol}{RGB}{217,83,25}
\definecolor{nullCol}{RGB}{126,47,142}
\definecolor{elsCol}{RGB}{237,177,32}

\begin{figure}[ht]
\centering
\begin{tikzpicture}
\begin{groupplot}[
    group style={
        group size=2 by 1,
        horizontal sep=0.2cm,
        y descriptions at=edge left,
    },
    width=0.575\columnwidth,
    height=0.5\columnwidth,
    grid=both,
    grid style={line width=0.4pt, draw=gray!30},
    major grid style={line width=0.6pt, draw=gray!50},
    every axis plot/.append style={line width=1pt, mark size=2pt},
    ymode=log,
    ytick={1e-3, 1e-1, 1e1, 1e3},
    ymin=1e-4, ymax=1e3,
    tick label style={font=\small},
    label style={font=\small},
    xlabel style={
        at={(axis description cs:0.5,0)},
        anchor=north,
        yshift=-0.35cm,
    },
]

\nextgroupplot[
    xmode=log,
    xlabel={Dictionary size $n$},
    ylabel={Average runtime (s)},
    legend to name=runtimelegend,
    legend columns=4,
    legend style={
        draw=none,
        font=\footnotesize,
        /tikz/every even column/.append style={column sep=0.6em},
    },
]
\addplot[color=bpCol, mark=o, mark options={fill=white}] table[x=n,y=bp_n]{figureRuntime.dat};
\addlegendentry{BP}
\addplot[color=ompCol, mark=square, mark options={fill=white}] table[x=n,y=omp_n]{figureRuntime.dat};
\addlegendentry{OMP}
\addplot[color=elsCol, mark=triangle, mark options={fill=white}] table[x=n,y=els_n]{figureRuntime.dat};
\addlegendentry{SESP-D}
\addplot[color=nullCol, mark=triangle, mark options={fill=white, rotate=180}] table[x=n,y=null_n]{figureRuntime.dat};
\addlegendentry{SESP-P}

\nextgroupplot[
    xtick={1,5,10,15,20},
    xlabel={Sparsity $s$},
]
\addplot[color=bpCol, mark=o, mark options={fill=white}, forget plot] table[x=s,y=bp_s]{figureRuntime.dat};
\addplot[color=ompCol, mark=square, mark options={fill=white}, forget plot] table[x=s,y=omp_s]{figureRuntime.dat};
\addplot[color=elsCol, mark=triangle, mark options={fill=white}, forget plot] table[x=s,y=els_s]{figureRuntime.dat};
\addplot[color=nullCol, mark=triangle, mark options={fill=white, rotate=180}, forget plot] table[x=s,y=null_s]{figureRuntime.dat};

\end{groupplot}

\node at ($(group c1r1.north)!0.5!(group c2r1.north)$)
    [anchor=south, yshift=0.05cm] {\ref{runtimelegend}};

\end{tikzpicture}

\vspace{-0.55cm}

\begin{minipage}[t]{0.15\columnwidth}
\centering
\end{minipage}%
\hspace{1.2cm}%
\begin{minipage}[t]{0.4\columnwidth}
\centering\small
\subcaption{Fixed sparsity $s=10$.}
\label{fig:runtime_n}
\end{minipage}%
\begin{minipage}[t]{0.4\columnwidth}
\centering\small
\subcaption{Fixed size $n=500$.}
\label{fig:runtime_s}
\end{minipage}

\caption{Average runtime comparison in the noiseless case.}
\label{fig:runtime_comparison}
\end{figure}
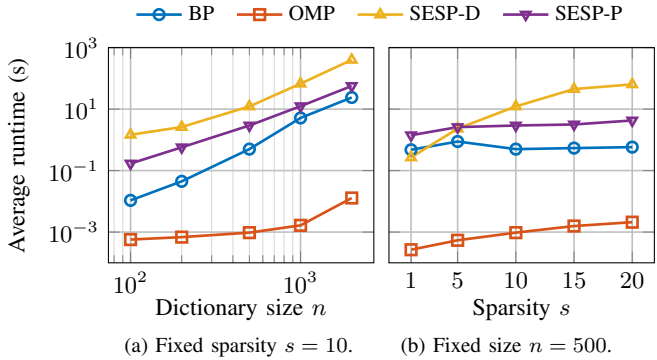

Fig.~\ref{fig:runtime_n} shows that BP, SESP-D, and SESP-P scale similarly with matrix dimension at fixed sparsity. Fig.~\ref{fig:runtime_s} shows that SESP-D runtime grows with sparsity for fixed $n$, while BP and SESP-P remain largely unaffected. The per-iteration cost of SESP-D is comparable to that of BP and SESP-P; the increase instead reflects more iterations needed to converge at higher sparsity levels.

\subsection{Recovery and Runtime in the Noisy Case}
\label{sec:noisy_experiments}

Fig.~\ref{fig:err_vs_snr} reports recovery performance under additive Gaussian noise, with 100 independent trials per SNR level and error bars indicating the 25th and 75th percentiles. The oracle estimator, which solves the least-squares problem restricted to the true support $S$, lower-bounds the achievable error for any method that must identify the support. Since SESP-D and SESP-P estimate the support before applying a support-restricted least-squares correction (Section~\ref{sec:implementation_algs}), the same debiasing step is applied to BPDN and OMP for a fair comparison.

Fig.~\ref{fig:err_vs_snr} shows that, unlike BP and OMP, SESP-D and SESP-P remain close to the oracle bound, demonstrating accurate support recovery even at higher sparsity levels. Initializing SESP-D and SESP-P with the BPDN solution yields essentially the same reconstruction accuracy as the other initialization strategies in Section~\ref{sec:implementation_algs}.


\definecolor{bpCol}{RGB}{0,114,189}
\definecolor{ompCol}{RGB}{217,83,25}
\definecolor{nullCol}{RGB}{126,47,142}
\definecolor{elsCol}{RGB}{237,177,32}

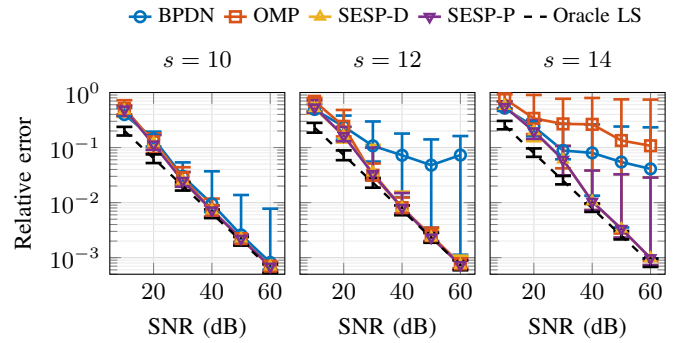
\begin{figure}[ht]
\centering
\begin{tikzpicture}
\begin{groupplot}[
    group style={
        group size=3 by 1,
        horizontal sep=0.2cm,
        y descriptions at=edge left,
    },
    width=0.44\columnwidth,
    height=0.45\columnwidth,
    ymode=log,
    ymin=5*1e-4, ymax=1e0,
    xmin=5, xmax=65,
    grid=both,
    minor grid style={gray!15},
    major grid style={gray!25},
    error bars/y dir=both,
    error bars/y explicit,
    error bars/error bar style={line width=1pt},
    error bars/error mark options={
        line width=6pt,
        mark size=0.5pt,
        },
    every axis plot/.append style={line width=1pt, mark size=2pt},
    label style={font=\small},
    tick label style={font=\small},
    title style={font=\small},
    xlabel={SNR (dB)},
    xtick={20,40,60},
    xticklabel style={font=\small},
    yticklabel style={font=\small},
    legend image code/.code={%
      \draw[mark repeat=2, mark phase=2, #1]
        plot coordinates {(0cm,0cm) (0.15cm,0cm) (0.3cm,0cm)};%
    },
]

\nextgroupplot[
    title={$s = 10$},
    ylabel={Relative error},
]
\addplot[bpCol,   mark=o, mark options={fill=white}]        table[x=snr, y=bp_med_k10,
    y error plus=bp_errhi_k10, y error minus=bp_errlo_k10] {figureErrorVsSNR.dat};

\addplot[ompCol,  mark=square, mark options={fill=white}]  table[x=snr, y=omp_med_k10,
    y error plus=omp_errhi_k10, y error minus=omp_errlo_k10] {figureErrorVsSNR.dat};

\addplot[elsCol,  mark=triangle, mark options={fill=white}] table[x=snr, y=els_med_k10,
    y error plus=els_errhi_k10, y error minus=els_errlo_k10] {figureErrorVsSNR.dat};

\addplot[nullCol, mark=triangle, mark options={fill=white, rotate=180}] table[x=snr, y=null_med_k10,
    y error plus=null_errhi_k10, y error minus=null_errlo_k10] {figureErrorVsSNR.dat};

\addplot[black, dashed, mark=none, line width=0.9pt] table[x=snr, y=oracle_med_k10,
    y error plus=oracle_errhi_k10, y error minus=oracle_errlo_k10] {figureErrorVsSNR.dat};

\nextgroupplot[
    title={$s = 12$},
    yticklabel=\empty,
    legend columns=5,
    legend style={
        at={(0.5,1.5)}, anchor=north,
        font=\footnotesize,
        draw=none,
        inner sep=0pt,
        fill=none,
        /tikz/every even column/.append style={column sep=1pt},
    },
    legend cell align=left,
]
\addplot[bpCol,   mark=o, mark options={fill=white}]        table[x=snr, y=bp_med_k12,
    y error plus=bp_errhi_k12, y error minus=bp_errlo_k12] {figureErrorVsSNR.dat};
\addlegendentry{BPDN}
\addplot[ompCol,  mark=square, mark options={fill=white}]  table[x=snr, y=omp_med_k12,
    y error plus=omp_errhi_k12, y error minus=omp_errlo_k12] {figureErrorVsSNR.dat};
\addlegendentry{OMP}
\addplot[elsCol,  mark=triangle, mark options={fill=white}] table[x=snr, y=els_med_k12,
    y error plus=els_errhi_k12, y error minus=els_errlo_k12] {figureErrorVsSNR.dat};
\addlegendentry{SESP-D}
\addplot[nullCol, mark=triangle, mark options={fill=white, rotate=180}] table[x=snr, y=null_med_k12,
    y error plus=null_errhi_k12, y error minus=null_errlo_k12] {figureErrorVsSNR.dat};
\addlegendentry{SESP-P}
\addplot[black, dashed, mark=none, line width=0.9pt] table[x=snr, y=oracle_med_k12,
    y error plus=oracle_errhi_k12, y error minus=oracle_errlo_k12] {figureErrorVsSNR.dat};
\addlegendentry{Oracle LS}

\nextgroupplot[
    title={$s = 14$},
    yticklabel=\empty,
]
\addplot[bpCol,   mark=o, mark options={fill=white}]        table[x=snr, y=bp_med_k14,
    y error plus=bp_errhi_k14, y error minus=bp_errlo_k14] {figureErrorVsSNR.dat};
\addplot[ompCol,  mark=square, mark options={fill=white}]  table[x=snr, y=omp_med_k14,
    y error plus=omp_errhi_k14, y error minus=omp_errlo_k14] {figureErrorVsSNR.dat};
\addplot[elsCol,  mark=triangle, mark options={fill=white}] table[x=snr, y=els_med_k14,
    y error plus=els_errhi_k14, y error minus=els_errlo_k14] {figureErrorVsSNR.dat};
\addplot[nullCol, mark=triangle, mark options={fill=white, rotate=180}] table[x=snr, y=null_med_k14,
    y error plus=null_errhi_k14, y error minus=null_errlo_k14] {figureErrorVsSNR.dat};
\addplot[black, dashed, mark=none, line width=0.9pt] table[x=snr, y=oracle_med_k14,
    y error plus=oracle_errhi_k14, y error minus=oracle_errlo_k14] {figureErrorVsSNR.dat};

\end{groupplot}
\end{tikzpicture}
\caption{Median relative error versus SNR. Error bars indicate the 25th and 75th percentiles over 100 trials. The oracle lower bound corresponds to least-squares on the true support.}
\label{fig:err_vs_snr}
\end{figure}

The improved robustness to measurement noise comes at increased computational cost: Fig.~\ref{fig:runtime_comparison_noisy} shows SESP-D and SESP-P requiring longer runtimes than BPDN and OMP, with a larger gap than in the noiseless case (Fig.~\ref{fig:runtime_comparison}).

\definecolor{bpCol}{RGB}{0,114,189}
\definecolor{ompCol}{RGB}{217,83,25}
\definecolor{sesppCol}{RGB}{126,47,142}
\definecolor{sespdCol}{RGB}{237,177,32}

\begin{figure}[ht]
\centering
\begin{tikzpicture}
\begin{groupplot}[
    group style={
        group size=2 by 1,
        horizontal sep=0.2cm,
        y descriptions at=edge left,
    },
    width=0.575\columnwidth,
    height=0.5\columnwidth,
    grid=both,
    grid style={line width=0.4pt, draw=gray!30},
    major grid style={line width=0.6pt, draw=gray!50},
    every axis plot/.append style={line width=1pt, mark size=2pt},
    ymode=log,
    ytick={1e-3, 1e-1, 1e1, 1e3},
    ymin=1e-4, ymax=1e3,
    tick label style={font=\small},
    label style={font=\small},
    xlabel style={
        at={(axis description cs:0.5,0)},
        anchor=north,
        yshift=-0.35cm,
    },
]

\nextgroupplot[
    xmode=log,
    xlabel={Dictionary size $n$},
    ylabel={Average runtime (s)},
    legend to name=runtimelegendnoisy,
    legend columns=4,
    legend style={
        draw=none,
        font=\footnotesize,
        /tikz/every even column/.append style={column sep=0.6em},
    },
]
\addplot[color=bpCol, mark=o, mark options={fill=white}] table[x=n,y=bp_n]{figureRuntimeNoisy.dat};
\addlegendentry{BPDN}
\addplot[color=ompCol, mark=square, mark options={fill=white}] table[x=n,y=omp_n]{figureRuntimeNoisy.dat};
\addlegendentry{OMP}
\addplot[color=sespdCol, mark=triangle, mark options={fill=white}] table[x=n,y=els_n]{figureRuntimeNoisy.dat};
\addlegendentry{SESP-D}
\addplot[color=sesppCol, mark=triangle, mark options={fill=white, rotate=180}] table[x=n,y=null_n]{figureRuntimeNoisy.dat};
\addlegendentry{SESP-P}

\nextgroupplot[
    xtick={1,5,10,15,20},
    xlabel={Sparsity $s$},
]
\addplot[color=bpCol, mark=o, mark options={fill=white}, forget plot] table[x=s,y=bp_s]{figureRuntimeNoisy.dat};
\addplot[color=ompCol, mark=square, mark options={fill=white}, forget plot] table[x=s,y=omp_s]{figureRuntimeNoisy.dat};
\addplot[color=sespdCol, mark=triangle, mark options={fill=white}, forget plot] table[x=s,y=els_s]{figureRuntimeNoisy.dat};
\addplot[color=sesppCol, mark=triangle, mark options={fill=white, rotate=180}, forget plot] table[x=s,y=null_s]{figureRuntimeNoisy.dat};

\end{groupplot}

\node at ($(group c1r1.north)!0.5!(group c2r1.north)$)
    [anchor=south, yshift=0.05cm] {\ref{runtimelegendnoisy}};

\end{tikzpicture}

\vspace{-0.55cm}

\begin{minipage}[t]{0.15\columnwidth}
\centering
\end{minipage}%
\hspace{1.2cm}%
\begin{minipage}[t]{0.4\columnwidth}
\centering\small
\subcaption{Fixed sparsity $s=10$}
\label{fig:runtime_noisy_n}
\end{minipage}%
\begin{minipage}[t]{0.4\columnwidth}
\centering\small
\subcaption{Fixed size $n=500$}
\label{fig:runtime_noisy_s}
\end{minipage}

\caption{Average runtime comparison in the noisy case.}
\label{fig:runtime_comparison_noisy}
\end{figure}
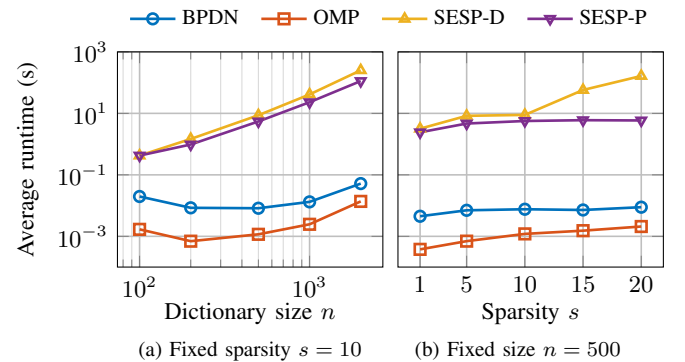

Runtimes for warm-starting SESP-D and SESP-P from the BPDN solution are omitted in Fig.~\ref{fig:runtime_comparison_noisy}, since BPDN alone already achieves sufficient accuracy for all $n$ and $s$ considered. In regimes where BPDN accuracy is insufficient, however, it can serve as an inexpensive initialization: a low-cost BPDN estimate is first computed, then refined with SESP-D or SESP-P, trading additional computation for improved accuracy.

\subsection{Sparse Modeling with Equivalent Representations}
\label{sec:sparse_modeling_experiment}

We next illustrate the ability of the proposed methods to recover multiple equivalent sparse representations of the same signal while explicitly controlling the target sparsity level.

The proposed methods may return sparse solutions with equal sparsity but different supports. Hence, performance is no longer evaluated via the relative error on $\hat{\mathbf{x}}$. Instead, a computed solution $\hat{\mathbf{x}}$ is declared successful if it satisfies $\|\mathbf{A}\hat{\mathbf{x}} - \mathbf{b}\|_2 < 10^{-10}$ and $\|\hat{\mathbf{x}}\|_0 \le s$, where entries with $|\hat{x}_i|<10^{-8}$ are treated as zero when computing $\|\hat{\mathbf{x}}\|_0$.

Let $\mathbf{t}$ be the time vector with entries $t_i=i$, for $i=0,1,\dots,10$. We construct
the matrix
\[
\mathbf{A} = \begin{bmatrix}
\mathbf{1} & \mathbf{t} & \mathbf{t}-\mathbf{1} & \mathbf{t}+\mathbf{1} & \mathbf{t}^{\cdot 2} & \mathbf{t}^{\cdot 3}
\end{bmatrix},
\]
where $\mathbf{t}^{\cdot k}$ denotes the entrywise power and $\mathbf{1}$ is the all-ones vector.
Since $\mathbf{A}$ has linearly dependent columns, the same signal may admit multiple sparse
representations with different supports and different sparsity levels.

Consider the signal $\mathbf{b}= \tfrac{2}{3}\cdot\mathbf{1} + \tfrac{1}{3}\cdot \mathbf{t}^{\cdot 3}$, which admits the
equivalent representations
\begin{align*}
\hat{\mathbf{x}}_1 &= \begin{bmatrix} \tfrac{2}{3}&0&0&0&0&\tfrac{1}{3} \end{bmatrix}^\top, 
&\|\hat{\mathbf{x}}_1\|_0 = 2, \\
\hat{\mathbf{x}}_2 &= \begin{bmatrix} 0&0&-\tfrac{1}{3}&\tfrac{1}{3}&0&\tfrac{1}{3} \end{bmatrix}^\top, 
&\|\hat{\mathbf{x}}_2\|_0 = 3, \\
\hat{\mathbf{x}}_3 &= \begin{bmatrix} 0&\tfrac{2}{3}&-\tfrac{2}{3}&0&0&\tfrac{1}{3} \end{bmatrix}^\top, 
&\|\hat{\mathbf{x}}_3\|_0 = 3, \\
\hat{\mathbf{x}}_4 &= \begin{bmatrix} 0&-\tfrac{2}{3}&0&\tfrac{2}{3}&0&\tfrac{1}{3} \end{bmatrix}^\top, 
&\|\hat{\mathbf{x}}_4\|_0 = 3.
\end{align*}
BP seeks a single minimum-$\ell_1$ solution and cannot find multiple sparse representations. OMP follows a deterministic greedy path and is similarly limited to a single solution. In contrast, SESP-D and SESP-P recover all four solutions.

SESP-D and SESP-P require specifying a target sparsity level $s$. The multilinear sparsity constraints define an algebraic variety containing all vectors with $\|\mathbf{x}\|_0 \le s$; consequently, even when the target sparsity exceeds the true sparsity, lower-sparsity solutions remain recoverable. In this example, setting $s=2$ always yields $\hat{\mathbf{x}}_1$, while setting $s=3$ yields $\hat{\mathbf{x}}_2$, $\hat{\mathbf{x}}_3$, or $\hat{\mathbf{x}}_4$, but also $\hat{\mathbf{x}}_1$, depending on the initialization.

SESP-D is initialized with a vector drawn from the normal distribution $\mathbf{x}_0 \sim \mathcal{N}(\mathbf{0}, \gamma^2 \mathbf{I}_n)$, where $\gamma > 0$. SESP-P is initialized with $\mathbf{y}_0 = \mathbf{V}^\top (\mathbf{x}_0-\mathbf{x}_p)$. We vary $\gamma$ to examine how initialization scale affects the recovery frequency of each of the four solutions, testing 1000 random initializations per value of $\gamma$ and reporting the resulting percentages in Fig.~\ref{fig:multiple_solutions}. The results show that even when the target sparsity $s$ overestimates the true minimal sparsity, the recovered solution is most often the one of smaller sparsity ($\hat{\mathbf{x}}_1$) rather than the higher-sparsity alternatives ($\hat{\mathbf{x}}_2$, $\hat{\mathbf{x}}_3$, $\hat{\mathbf{x}}_4$).

\begin{figure}[ht]
\centering
\begin{tikzpicture}
\begin{groupplot}[
    group style={
        group size=2 by 1,
        horizontal sep=1cm,
    },
    ybar stacked,
    width=0.52\columnwidth, height=0.37\columnwidth,
    ymin=0, ymax=1,
    symbolic x coords={10,3,1,0.3,0.1},
    xtick=data,
    xlabel={$\gamma$},
    tick label style={font=\small},
    label style={font=\small},
    title style={font=\footnotesize},
    ]

\nextgroupplot[
    title={SESP-D},
    ylabel={\shortstack{Recovery\\frequency}},
    ylabel style={align=center},
    legend to name={fig:recovery_freq_legend},
    legend style={legend columns=1, font=\footnotesize, draw=none},
    ]
\addplot coordinates {(10,0.541) (3,0.529) (1,0.518) (0.3,0.550) (0.1,0.602)};
\addplot coordinates {(10,0.248) (3,0.294) (1,0.344) (0.3,0.340) (0.1,0.374)};
\addplot coordinates {(10,0.108) (3,0.082) (1,0.066) (0.3,0.043) (0.1,0.010)};
\addplot coordinates {(10,0.103) (3,0.095) (1,0.072) (0.3,0.067) (0.1,0.014)};
\legend{$\hat{\mathbf{x}}_1$,$\hat{\mathbf{x}}_2$,$\hat{\mathbf{x}}_3$,$\hat{\mathbf{x}}_4$}

\nextgroupplot[
    title={SESP-P},
    yticklabels={},
    ]
\addplot coordinates {(10,0.467) (3,0.425) (1,0.481) (0.3,0.514) (0.1,0.526)};
\addplot coordinates {(10,0.213) (3,0.209) (1,0.318) (0.3,0.486) (0.1,0.474)};
\addplot coordinates {(10,0.162) (3,0.186) (1,0.108) (0.3,0.0)   (0.1,0.0)};
\addplot coordinates {(10,0.158) (3,0.180) (1,0.093) (0.3,0.0)   (0.1,0.0)};

\end{groupplot}

\node[anchor=center] at ($(group c1r1.east)!0.5!(group c2r1.west)$)
    {\ref{fig:recovery_freq_legend}};

\end{tikzpicture}
\caption{Recovery frequency by component for SESP-D (left) and SESP-P (right) as a function of $\gamma$.}
\label{fig:multiple_solutions}
\end{figure}
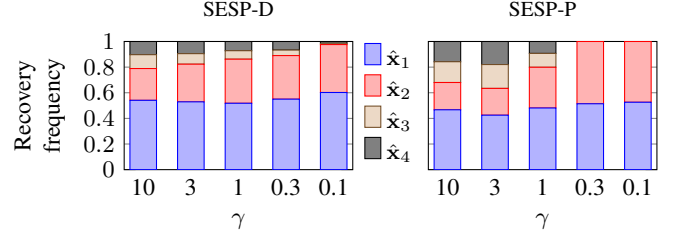

The Macaulay approach is also evaluated on this example. Unlike the previous experiments, where its computational cost rendered it intractable, the limited problem size here makes it feasible. In contrast to SESP-D and SESP-P, the Macaulay construction is formally guaranteed to recover all solutions at once. Specifically, $\hat{\mathbf{x}}_2$, $\hat{\mathbf{x}}_3$, and $\hat{\mathbf{x}}_4$ are recovered each with multiplicity one, and $\hat{\mathbf{x}}_1$ with multiplicity three.

The tensor constructed in \eqref{eq:cpd_null} is decomposed using the CPD-NLS algorithm in Tensorlab with GEVD-based initialization. Function and step tolerances are set to $\varepsilon^2$ and $\varepsilon$, respectively, and a maximum of $500$ iterations is used. The solutions at infinity that the method finds are discarded.

\section{Conclusion and Future Work}
We reformulated sparse approximation as a structured system of polynomial equations. The Macaulay construction is guaranteed to algebraically recover all $s$-sparse solutions at once, but is tractable only for small problems. We also proposed two optimization-based methods, SESP-D and \mbox{SESP-P}, which exploit the structure of the sparsity constraints {via inexpensive ESP estimation and} keep the per-iteration cost comparable to that of BP.

In our experiments, SESP-D and SESP-P recovered sparse signals from fewer measurements than BP and OMP, succeeded at higher sparsity levels, and remained close to the oracle bound under additive noise. This accuracy comes at a higher computational cost, which is most pronounced in the noisy setting. When BP or BPDN already succeeds, it remains the cheaper choice; when it does not, its solution can serve as an inexpensive initialization for the proposed methods.

The formulation offers two further properties {of interest}. First, several equivalent sparse representations of the same signal can be recovered, which neither BP nor OMP provides. Second, unlike in BP and BPDN, where sparsity is obtained indirectly through a regularization parameter, the target sparsity is specified directly and acts as an upper bound rather than an exact requirement.

Several directions remain open. First, the connection with EVD and CPD makes it in principle possible to derive upper bounds on the estimation error in the noisy case. Second, our approach extends to constrained sparse approximation, since any constraint that can be written polynomially can be appended to the system. Examples of such constrained variants are: the unit-norm constraint, structured sparsity patterns, and nonnegativity. Third, the approach may be extended to parametrized dictionaries if the dependence on the unknown parameters is polynomial, as those parameters may be treated as additional variables. Fourth, faster Macaulay solvers would make the simultaneous recovery of all solutions practical for larger problems.

\appendices

\section{Structure of $\mathbf{Z}^\top \mathbf{Z}$}
\label{appx:structure_M}
 
Let $[n] = \{1,2,\dots,n\}$. We index the standard basis of $\mathbb{R}^{n^{s+1}}$ by multi-indices
$(i_1,\dots,i_{s+1}) \in [n]^{s+1},$
with corresponding basis vectors
$\mathbf{e}_{i_1,\dots,i_{s+1}} = \mathbf{e}_{i_1} \kron \mathbf{e}_{i_2} \kron \cdots \kron \mathbf{e}_{i_{s+1}}$
where $\mathbf{e}_{i_j}$ are basis vectors of $\mathbb{R}^n.$

Since $\mathbf{Z}$ consists of one-hot rows, the matrix $\mathbf{Z}^\top \mathbf{Z}$ is diagonal with diagonal entries in $\{0,1\}$.
The matrix $\mathbf{Z}^\top \mathbf{Z}$ acts as a projector onto the subspace of $\mathbb{R}^{n^{s+1}}$ spanned by basis vectors whose indices are pairwise distinct. Hence
\[
\mathbf{Z}^\top \mathbf{Z}
=
\sum_{\substack{(i_1,\dots,i_{s+1}) \in [n]^{s+1}\\
i_1 \neq \cdots \neq i_{s+1}}}
\mathbf{e}_{i_1,\dots,i_{s+1}}
\mathbf{e}_{i_1,\dots,i_{s+1}}^\top .
\]
Equivalently, we obtain the Kronecker representation
\begin{equation}
\label{eq:structure_MTM}
\mathbf{Z}^\top \mathbf{Z}
=
\sum_{\substack{(i_1,\dots,i_{s+1}) \in [n]^{s+1}\\
i_1 \neq \cdots \neq i_{s+1}}}
\mathbf{e}_{i_1} \mathbf{e}_{i_1}^\top
\kron
\cdots
\kron
\mathbf{e}_{i_{s+1}} \mathbf{e}_{i_{s+1}}^\top .
\end{equation}

\section{SESP-D Derivation}
\label{appx:elscpd_derivation}

{\emph{Objective Function}}
The objective function $f$ is obtained from the residual~$\mathbf{r}$ as $f = \frac{1}{2} \mathbf{r}^\top \mathbf{r}$.
From the formulation of SESP-D~\eqref{eq:elscpd_formulation}, it follows
\begin{equation}
\label{eq:residual}
    \mathbf{r} = \begin{bmatrix} \mathbf{r}^{(1)} \\ \lambda \mathbf{r}^{(2)} \end{bmatrix} = \begin{bmatrix} \mathbf{Ax} - \mathbf{b} \\ \lambda \mathbf{Z} \left(\kron_{i=1}^{s+1} \mathbf{x} \right) \end{bmatrix}.
\end{equation}
Thus, $f = \frac{1}{2} \| \mathbf{Ax} - \mathbf{b} \|_2^2 + \frac{1}{2} \lambda^2 \mathbf{r}^{(2)\top} \mathbf{r}^{(2)}$.

Using the mixed-product property of the Kronecker product, we obtain
$(\mathbf{x} \kron \mathbf{x})^\top (\mathbf{e}_i \mathbf{e}_i^\top \kron \mathbf{e}_j \mathbf{e}_j^\top) (\mathbf{x} \kron \mathbf{x}) =
(\mathbf{x} ^\top \mathbf{e}_i \mathbf{e}_i^\top\mathbf{x}) \kron (\mathbf{x} ^\top \mathbf{e}_j \mathbf{e}_j^\top\mathbf{x}) 
=  x_i^2 x_j^2$
Combining this identity with the structure of $\mathbf{Z}^\top \mathbf{Z}$ in~\eqref{eq:structure_MTM} yields
\begin{align*}
\mathbf{r}^{(2)\top} \mathbf{r}^{(2)}
&=
\lambda^2 \left(\kron_{k=1}^{s+1} \mathbf{x}\right)^\top
\mathbf{Z}^\top \mathbf{Z}
\left(\kron_{l=1}^{s+1} \mathbf{x}\right) \\
&=
\lambda^2 \sum_{\substack{(i_1,\dots,i_{s+1}) \in [n]^{s+1}\\
i_1 \neq \ldots \neq i_{s+1} }}
\prod_{k=1}^{s+1} x_{i_k}^2 .
\end{align*}

Each ordered tuple $(i_1,\dots,i_{s+1})$ with distinct elements corresponds to a permutation of a subset $S \subset [n]$ with $|S| = s+1$. For each subset there are $(s+1)!$ such permutations. Grouping terms by subsets therefore yields
\begin{equation*}
\mathbf{r}^{(2)\top} \mathbf{r}^{(2)}
= (s+1)! \sum_{\substack{S \subset [n] \\ |S|=s+1}} \prod_{i \in S} x_i^2 .
\end{equation*}

Expanding the sum over subsets gives
\begin{align*}
\sum_{\substack{S \subset [n] \\ |S|=s+1}}
\prod_{i \in S} x_i^2
&=
\sum_{ \substack{1 \le i_1 <\cdots \\ \cdots < i_{s+1} \le n}}
x_{i_1}^2 \cdots x_{i_{s+1}}^2 \overset{\text{Def.~\ref{def:elem_sym_poly}}}{=}
E_{s+1}.
\end{align*}
Therefore, the expression for the objective function is
\begin{equation*}
f
=
\frac{1}{2}\|\mathbf{Ax}-\mathbf{b}\|_2^2
+
\frac{1}{2}\lambda^2(s+1)! \,
E_{s+1}.
\end{equation*}

{\emph{Gradient:}}
The gradient splits as
\[
\mathbf{J}^\top \mathbf{r} = \nabla f =
\frac12 \nabla \| \mathbf{Ax} - \mathbf{b} \|_2^2 
+ \frac12 \lambda^2 (s+1)! \nabla E_{s+1} 
\]
where $\nabla \| \mathbf{Ax} - \mathbf{b} \|_2^2  = 2 \mathbf{A}^\top (\mathbf{Ax} - \mathbf{b})$, and $\nabla E_{s+1}$ is obtained using the fact that $ \frac{\partial}{\partial x_i} E_{k+1}= 2x_i E_k^{(i)}$. Thus, $\nabla E_{s+1} = 2 \sum_{i=1}^n E_s^{(i)} x_i \mathbf{e}_i $.

{\emph{Gramian of the Jacobians:}}
The Gramian of the Jacobians splits as
\begin{equation*}
\mathbf{J}^\top \mathbf{J}
=
\mathbf{A}^\top \mathbf{A}
+ \lambda^2 \mathbf{J}^{(2)\top}\mathbf{J}^{(2)},
\end{equation*}
where $\mathbf{J}^{(1)} = \mathbf{A}$ follows from $\mathbf{r}^{(1)} = \mathbf{Ax} - \mathbf{b}$,
and $\mathbf{J}^{(2)}$ denotes the Jacobian of
$\mathbf{r}^{(2)}$~\eqref{eq:residual}. We derive $\mathbf{J}^{(2)\top}\mathbf{J}^{(2)}$ as follows.

The entries of $\mathbf{r}^{(2)}$ are indexed by ordered tuples
$(i_1,\ldots,i_{s+1}) \in [n]^{s+1}$ with pairwise distinct indices, and the
corresponding entry is $r_{(i_1,\ldots,i_{s+1})}(\mathbf{x}) = x_{i_1} \cdots x_{i_{s+1}}$. Every permutation of $S \subset [n]$ with $|S| = s+1$ yields the same value
\begin{equation*}
r_S(\mathbf{x}) = \prod_{i \in S} x_i,
\end{equation*}
so each subset $S$ contributes exactly $(s+1)!$ entries to $\mathbf{r}^{(2)}$,
all equal to $r_S(\mathbf{x})$ and hence all sharing the same gradient
$\nabla r_S$.
Using $\mathbf{J}^{(2)\top}\mathbf{J}^{(2)} = \sum_k (\nabla r_k^{(2)})(\nabla r_k^{(2)})^\top$
and grouping terms by subset gives
\begin{equation}
\label{eq:JTJ2}
\mathbf{J}^{(2)\top}\mathbf{J}^{(2)}
=
(s+1)! \sum_{\substack{S \subset [n] \\ |S|=s+1}} (\nabla r_S)(\nabla r_S)^\top.
\end{equation}
Since $\nabla r_S = \sum_{i \in S} \prod_{k \in S \setminus \{i\}} x_k \mathbf{e}_i,$, it follows that
\begin{equation*}
(\nabla r_S)(\nabla r_S)^\top = \sum_{i,j \in S} \prod_{k \in S \setminus \{i\}} x_k \prod_{l \in S \setminus \{j\}} x_l \mathbf{e}_i \mathbf{e}_j^\top.
\end{equation*}

If $i = j$, the coefficient simplifies to $\prod_{k \in S \setminus \{i\}} x_k^2$. Summing over all subsets $S$ in \eqref{eq:JTJ2} yields $E_s^{(i)}$ at the position $\mathbf{e}_i \mathbf{e}_i^\top$ in the matrix $\mathbf{J}^{(2)\top}\mathbf{J}^{(2)}$.

If $i \neq j$, the coefficient factorises as $x_i x_j \prod_{k \in S \setminus \{i,j\}} x_k^2$. Summing over all subsets $S$ in \eqref{eq:JTJ2} yields $x_i x_j E_{s-1}^{(i,j)}$ at the positions $\mathbf{e}_i \mathbf{e}_j^\top$ and $\mathbf{e}_j \mathbf{e}_i^\top$ in the matrix $\mathbf{J}^{(2)\top}\mathbf{J}^{(2)}$.

Therefore,
\begin{equation*}
\mathbf{J}^{(2)\top}\mathbf{J}^{(2)}
= (s+1)!
\sum_{i=1}^n
E_s^{(i)}\, \mathbf{e}_i \mathbf{e}_i^\top
+
\sum_{\substack{i,j=1 \\ i \neq j}}^n
E_{s-1}^{(i,j)}\, x_i x_j\,
\mathbf{e}_i \mathbf{e}_j^\top.
\end{equation*}

\section{SESP-P Derivation}
\label{appx:nullspace_derivation}
The residual
\(
\mathbf{r} = \mathbf{Z} \left( \kron_{i=1}^{s+1} \mathbf{x} \right) = \mathbf{Z} \left( \kron_{i=1}^{s+1} (\mathbf{x}_p + \mathbf{Vy}) \right)
\)
is of the same form as \( \mathbf{r}^{(2)} \) in \eqref{eq:residual}. Thus, the objective
\begin{equation*}
f = \frac{1}{2} (s+1)! \, E_{s+1}
\end{equation*}
is obtained in the same way as in Appendix~\ref{appx:elscpd_derivation}.

An important difference is that SESP-P optimizes over \( \mathbf{y} \) instead of \( \mathbf{x} \). Since $\mathbf{x} = \mathbf{x}_p + \mathbf{V}\mathbf{y}$, the Jacobian of $\mathbf{x}$ with respect to $\mathbf{y}$ is $\frac{\partial \mathbf{x}}{\partial \mathbf{y}} = \mathbf{V}$.
Thus, using the chain rule,
\begin{align*}
    \mathbf{J}^\top \mathbf{r} &=
    \left(\frac{\partial \mathbf{x}}{\partial \mathbf{y}}\right)^\top \mathbf{J}^\top_\mathbf{x} \mathbf{r}
    =(s+1)! \, \mathbf{V}^\top \sum_{i=1}^n E_s^{(i)} x_i  \mathbf{e}_i, \\
    \mathbf{J}^\top \mathbf{J} &= (s+1)! \mathbf{V}^\top \Big( \sum_i E_s^{(i)} \mathbf{e}_i \mathbf{e}_i^\top + \sum_{i,j} E_{s-1}^{(i,j)} x_i x_j  \mathbf{e}_i \mathbf{e}_j^\top \Big) \mathbf{V}.
\end{align*}

\section{Proof of Lemma~\ref{lemma:esp_recursion}}
\label{appx:esp_recursion}
\begin{proof}
    Partition $e_d(z_1,\ldots, z_k)$, with $d \leq k$, into two parts: one that does not contain $z_k$ and one that does. The former is $e_d(z_1,\ldots,z_{k-1})=\sum_{1 \le i_1 < \cdots < i_{d} \le k-1} z_{i_1} \cdots z_{i_{d}}$. The latter is $z_k e_{d-1}(z_1,\ldots,z_{k-1})= $ $\sum_{1 \le i_1 < \cdots < i_{d-1} \le k-1} z_{i_1} \cdots z_{i_{d-1}} z_k$.
    Therefore, for every $ j = 1, \ldots, \min(k,d) $,
    \begin{equation}
    \label{eq:recursion}
        e_j(z_1,\ldots, z_k)= e_j(z_1,\ldots,z_{k-1}) + z_k e_{j-1}(z_1,\ldots,z_{k-1}).
    \end{equation}
    The recursion is initialized by $e_0(z_1,\ldots z_k) = 1$. Evaluating the recursion for $k=2,\ldots,n$ requires $\mathcal{O}(nd)$ operations and yields $e_d(z_1,\ldots,z_n)$.
\end{proof}

\section{Proof of Lemma~\ref{lemma:leave_one_out}}
\label{appx:leave_one_out}
\begin{proof}
    Compute left and right partial ESPs, for $k = 0, \ldots, s,$
    \begin{align*}
        E_k^{\mathrm{L}}(i) := e_k\!\left(x_1^2, \ldots, x_i^2\right), \quad
        E_k^{\mathrm{R}}(i) := e_k\!\left(x_i^2, \ldots, x_n^2\right), 
    \end{align*}
    with boundary conventions $E_0^{\mathrm{L}}(0) = 1$ and $E_0^{\mathrm{R}}(n+1) = 1$. Both tables are computed via~\eqref{eq:recursion}, sweeping left-to-right and right-to-left respectively, each in $O(ns)$ time.
    
    Partition each size-$s$ subset of $\{1,\ldots,n\}\setminus\{i\}$ according to how many elements lie strictly to the left versus strictly to the right of $i$. A subset contributing $j$ elements from $\{1,\ldots,i-1\}$ and $s-j$ elements from $\{i+1,\ldots,n\}$ contributes a term $E_j^{\mathrm{L}}(i-1)\, E_{s-j}^{\mathrm{R}}(i+1)$ to $E_s^{(i)}$. Thus, summing over $j$ gives
    \begin{equation*}
        E_s^{(i)} = \sum_{j=0}^{s} E_j^{\mathrm{L}}(i-1)\, E_{s-j}^{\mathrm{R}}(i+1),
    \end{equation*}
    which costs $O(s)$. Evaluating all $n$ values costs $O(ns)$. Together with the $O(ns)$ left and right table construction, the total complexity is $O(ns)$.
\end{proof}

\section{Proof of Lemma~\ref{lemma:leave_two_out}}
\label{appx:leave_two_out}
\begin{proof}
Fix $i \in \{1, \ldots, n\}$ and let $S^{(i)} := \{x^2_1, \ldots, x^2_n\} \setminus \{x^2_i\}$ denote the set of $n-1$ variables with $x^2_i$ removed. By Lemma~\ref{lemma:leave_one_out} applied to $S^{(i)}$, all $n-1$ values 
$E_{s-1}^{(i,j)}$, $j \neq i$, can be computed simultaneously in $\mathcal{O}(ns)$ time. Repeating for each $i = 1, \ldots, n$ yields all $\binom{n}{2}$ values $E_{s-1}^{(i,j)}$ in $\mathcal{O}(n^2 s)$ time in total.
\end{proof}

\section*{Acknowledgment}

This work was supported by the Flemish Government's AI Research Program and KU Leuven Internal Funds (iBOF/23/064, C14/22/096).

\ifCLASSOPTIONcaptionsoff
  \newpage
\fi

\bibliographystyle{IEEEtran}  
\bibliography{IEEEabrv,strings,references}

@STRING{IEEE_J_SP         = "{IEEE} Trans. Signal Process."}

@STRING{IEEE_J_IT         = "{IEEE} Trans. Inf. Theory"}

@STRING{IEEE_J_IP         = "{IEEE} Trans. Image Process."}

@STRING{IEEE_J_PAMI       = "{IEEE} Trans. Pattern Anal. Mach. Intell."}

@article{macaulay_poly,
    author = {Vanderstukken, Jeroen and De Lathauwer, Lieven},
    title = {Systems of Polynomial Equations, Higher-order Tensor Decompositions, and Multidimensional Harmonic Retrieval: A Unifying Framework. {Part I}: The Canonical Polyadic Decomposition},
    journal = J_SIAM_MAA,
    volume = {42},
    number = {2},
    pages = {883-912},
    year = {2021},
    doi = {10.1137/17M1150050}}

@article{macaulay_matrix,
author = {Kim Batselier and Philippe Dreesen and Bart {De Moor}},
title = {On the null spaces of the {Macaulay} matrix},
journal = J_LAA,
volume = {460},
pages = {259-289},
year = {2014},
issn = {0024-3795},
doi = {https://doi.org/10.1016/j.laa.2014.07.035}}

@article{bousse_LS_CPD,
    author = {Boussé, M. and Vervliet, N. and Domanov, I. and Debals, O. and De Lathauwer, L.},
    title = {Linear systems with a canonical polyadic decomposition constrained solution: Algorithms and applications},
    journal = J_NLAA,
    volume = {25},
    number = {6},
    pages = {1--18},
    year = {2018},
    month = {Aug.},
    doi = {https://doi.org/10.1002/nla.2190}}

@article{kolda_overview,
    author = {Kolda, Tamara G. and Bader, Brett W.},
    title = {Tensor Decompositions and Applications},
    journal = J_SIAM_REV,
    volume = {51},
    number = {3},
    pages = {455-500},
    year = {2009},
    doi = {10.1137/07070111X},
    eprint = { https://doi.org/10.1137/07070111X}}

@article{nikos_overview,
    author={Sidiropoulos, Nicholas D. and De Lathauwer, Lieven and Fu, Xiao and Huang, Kejun and Papalexakis, Evangelos E. and Faloutsos, Christos},
    journal=IEEE_J_SP, 
    title={Tensor Decomposition for Signal Processing and Machine Learning}, 
    year={2017},
    volume={65},
    number={13},
    pages={3551-3582},
    doi={10.1109/TSP.2017.2690524}}

@article{hendrikx2022block,
  title={Block row Kronecker-structured linear systems with a low-rank tensor solution},
  author={Hendrikx, Stijn and De Lathauwer, Lieven},
  journal=J_FAMS,
  volume={8},
  pages={832883},
  year={2022},
  publisher={Frontiers Media SA}
}

@incollection{nico_numerical,
  author = {Vervliet, N. and De Lathauwer, L.},
  title={Numerical optimization-based algorithms for data fusion},
  booktitle={Data handling in science and technology},
  volume={31},
  pages={81--128},
  year={2019},
  publisher={Elsevier}
}

@book{elem_sym_poly,
    author = {Macdonald, I G},
    isbn = {9780198534891},
    title = {Symmetric Functions and Hall Polynomials},
    publisher = {Oxford University Press},
    year = {1995},
}

@Electronic{tensorlab3.0,
  Title                    = {Tensorlab 3.0},
  Author                   = {Vervliet, N. and Debals, O. and Sorber, L. and {Van Barel}, M. and {De Lathauwer}, L.},
  Year                     = {2016},
  Url                      = {https://www.tensorlab.net}
}

@book{numerical_opt,
  address = {New York, NY},
  author = {Nocedal, {Jorge} and Wright, {Stephen J.}},
  edition = {2. ed.},
  isbn = {978-0-387-30303-1},
  pagetotal = {XXII, 664},
  ppn_gvk = {502988711},
  publisher = {Springer},
  series = {Springer series in operations research and financial engineering},
  title = {Numerical optimization},
  year = 2006}

@article{lasso,
    ISSN = {00359246},
    author = {Robert Tibshirani},
    journal = J_JRSSB,
    number = {1},
    pages = {267--288},
    publisher = {[Royal Statistical Society, Oxford University Press]},
    title = {Regression Shrinkage and Selection via the Lasso},
    urldate = {2025-06-01},
    volume = {58},
    year = {1996}
}

@article{sparse_np_hard1,
author = {Natarajan, B. K.},
title = {Sparse Approximate Solutions to Linear Systems},
journal = J_SIAM_COMP,
volume = {24},
number = {2},
pages = {227-234},
year = {1995},
doi = {10.1137/S0097539792240406}}

@article{sparse_np_hard2,
title = {Sparse approximation is provably hard under coherent dictionaries},
journal = J_JCSS,
volume = {84},
pages = {32-43},
year = {2017},
issn = {0022-0000},
doi = {https://doi.org/10.1016/j.jcss.2016.07.001},
author = {A. Çivril}
}

@article{basis_pursuit,
    author = {Chen, Scott Shaobing and Donoho, David L. and Saunders, Michael A.},
    title = {Atomic Decomposition by Basis Pursuit},
    journal = J_SIAM_SC,
    volume = {20},
    number = {1},
    pages = {33-61},
    year = {1998},
    doi = {10.1137/S1064827596304010}}

@book{omp1950,
    author    = {Alan J. Miller},
    title     = {Subset Selection in Regression},
    edition   = {2nd},
    year      = {2002},
    publisher = {Chapman and Hall},
    address   = {London, U.K.}}

@article{omp2,
    author    = {G. Davis and S. Mallat and M. Avellaneda},
    title     = {Adaptive greedy approximations},
    journal   = J_CA,
    year      = {1997},
    volume    = {13},
    number    = {1},
    pages     = {57--98},
    doi       = {10.1007/BF02678430},
    issn      = {1432-0940}}

@INPROCEEDINGS{omp1,
    author={Pati, Y.C. and Rezaiifar, R. and Krishnaprasad, P.S.},
    booktitle={Proceedings of 27th Asilomar Conference on Signals, Systems and Computers}, 
    title={Orthogonal matching pursuit: recursive function approximation with applications to wavelet decomposition}, 
    year={1993},
    volume={},
    number={},
    pages={40-44 vol.1},
    doi={10.1109/ACSSC.1993.342465}}

@ARTICLE{OMP_Tropp_Gilbert,
  author={Tropp, Joel A. and Gilbert, Anna C.},
  journal=IEEE_J_IT, 
  title={Signal Recovery From Random Measurements Via Orthogonal Matching Pursuit}, 
  year={2007},
  volume={53},
  number={12},
  pages={4655-4666},
  doi={10.1109/TIT.2007.909108}}

@misc{spgl1,
   author = {E. van den Berg and M. P. Friedlander},
   title = {{SPGL1}: A solver for large-scale sparse reconstruction},
   note = {https://friedlander.io/spgl1},
   month = {December},
   year = 2019
}

@article{spgl2,
  Author = {E. van den Berg and M. P. Friedlander},
  Title = {Probing the Pareto frontier for basis pursuit solutions},
  year = {2008},
  journal = J_SIAM_SC,
  volume = {31},
  number = {2},
  pages = {890-912},
  doi = {10.1137/080714488}
}

@book{Cox1998,
  author    = {Cox, David A. and Little, John B. and O'Shea, Donal},
  title     = {Using Algebraic Geometry},
  series    = {Graduate Texts in Mathematics},
  volume    = {185},
  publisher = {Springer},
  address   = {New York, NY, USA},
  year      = {1998}
}

@book{elad,
    author = {Elad, Michael.},
    title = {Sparse and Redundant Representations : From Theory to Applications in Signal and Image Processing },
    year = {2010},
    address = {New York, NY},
    booktitle = {Sparse and Redundant Representations : From Theory to Applications in Signal and Image Processing},
    edition = {1st ed. 2010.},
    isbn = {1-4419-7011-8},
    language = {eng},
    publisher = {Springer New York}}

@book{FoucartRauhut2013,
  author    = {Simon Foucart and Holger Rauhut},
  title     = {A Mathematical Introduction to Compressive Sensing},
  publisher = {Birkh{\"a}user},
  address   = {New York, NY},
  year      = {2013},
  series    = {Applied and Numerical Harmonic Analysis},
  doi       = {10.1007/978-0-8176-4948-7},
  isbn      = {978-0-8176-4947-0},
  edition   = {1},
}

@ARTICLE{donoho,
  author={Donoho, D.L.},
  journal=IEEE_J_IT, 
  title={Compressed sensing}, 
  year={2006},
  volume={52},
  number={4},
  pages={1289-1306},
  doi={10.1109/TIT.2006.871582}}

@ARTICLE{denoising,
  author={Elad, Michael and Aharon, Michal},
  journal=IEEE_J_IP, 
  title={Image Denoising Via Sparse and Redundant Representations Over Learned Dictionaries}, 
  year={2006},
  volume={15},
  number={12},
  pages={3736-3745},
  doi={10.1109/TIP.2006.881969}}

@article{dict_learning,
author = {Mairal, Julien and Bach, Francis and Ponce, Jean and Sapiro, Guillermo},
title = {Online Learning for Matrix Factorization and Sparse Coding},
year = {2010},
issue_date = {3/1/2010},
publisher = {JMLR.org},
volume = {11},
issn = {1532-4435},
journal = J_JMLR,
month = mar,
pages = {19–60},
numpages = {42}
}

@ARTICLE{ksvd,
  author={Aharon, M. and Elad, M. and Bruckstein, A.},
  journal=IEEE_J_SP, 
  title={K-{SVD}: An algorithm for designing overcomplete dictionaries for sparse representation}, 
  year={2006},
  volume={54},
  number={11},
  pages={4311-4322}}

@ARTICLE{face_recog,
  author={Wright, John and Yang, Allen Y. and Ganesh, Arvind and Sastry, S. Shankar and Ma, Yi},
  journal=IEEE_J_PAMI, 
  title={Robust Face Recognition via Sparse Representation}, 
  year={2009},
  volume={31},
  number={2},
  pages={210-227},
  doi={10.1109/TPAMI.2008.79}}

@ARTICLE{super_resolution,
  author={Yang, Jianchao and Wright, John and Huang, Thomas S. and Ma, Yi},
  journal=IEEE_J_IP, 
  title={Image Super-Resolution Via Sparse Representation}, 
  year={2010},
  volume={19},
  number={11},
  pages={2861-2873},
  doi={10.1109/TIP.2010.2050625}}

@article{sindy,
author = {Steven L. Brunton  and Joshua L. Proctor  and J. Nathan Kutz },
title = {Discovering governing equations from data by sparse identification of nonlinear dynamical systems},
journal = J_PNAS,
volume = {113},
number = {15},
pages = {3932-3937},
year = {2016},
doi = {10.1073/pnas.1517384113}}

@book{nla_trefethen,
author = {Trefethen, Lloyd N. and Bau, III, David},
title = {Numerical Linear Algebra},
publisher = {Society for Industrial and Applied Mathematics},
year = {1997},
doi = {10.1137/1.9780898719574},
address = {Philadelphia, PA},
edition   = {},
eprint = {https://epubs.siam.org/doi/pdf/10.1137/1.9780898719574}
}

@STRING{J_SIAM_MAA  = "SIAM J. Matrix Anal. Appl."}

@STRING{J_LAA       = "Linear Algebra Appl."}

@STRING{J_NLAA      = "Numer. Linear Algebra Appl."}

@STRING{J_SIAM_REV  = "SIAM Rev."}

@STRING{J_FAMS      = "Front. Appl. Math. Stat."}

@STRING{J_JRSSB     = "J. R. Stat. Soc. Ser. B"}

@STRING{J_SIAM_COMP = "SIAM J. Comput."}

@STRING{J_JCSS      = "J. Comput. Syst. Sci."}

@STRING{J_SIAM_SC   = "SIAM J. Sci. Comput."}

@STRING{J_CA        = "Constr. Approx."}

@STRING{J_PNAS      = "Proc. Natl. Acad. Sci. U.S.A."}

@STRING{J_JMLR      = "J. Mach. Learn. Res."}

\end{document}